\documentclass{article}[11pt,letterpaper]
\usepackage{fullpage}
\newcommand{\titlenote}{\footnote}
\newcommand{\numberofauthors}[1]{}
\newcommand{\alignauthor}{\\}
\newcommand{\email}{\texttt}
\newcommand{\affaddr}[1]{#1}
\usepackage{amsthm}
\newcommand{\compress}[1]{#1}
\usepackage{enumitem} \setlist{noitemsep,topsep=0pt} \setitemize{noitemsep,topsep=0pt} \setenumerate{noitemsep,topsep=0pt}
\usepackage[latin1]{inputenc}
\usepackage{algorithmic}
\usepackage{algorithm}
\usepackage{amssymb,amsmath}
\usepackage{qtree}
\usepackage{acronym}

\usepackage{graphicx}
\usepackage[usenames,dvipsnames]{xcolor}
\usepackage{tkz-graph}
\usepackage{tikz}

\newcommand{\nega}[1]{1\text{ -- }#1}
\acrodef{FCNS}[FCNS]{First-Child-Next-Sibling}
\acrodef{DTD}[DTD]{Document Type Definition}
\acrodef{XML}[XML]{Extended Markup Language}
\newcommand{\pb}{\textsc{Validity}}
\newcommand{\disj}{\textsc{DISJ}}

\newcommand{\Order}{\mathrm{O}}

\newcommand{\eps}{\varepsilon}

\newcommand{\rev}{\mathrm{\mathrm{rev}}}

\newcounter{chriscom}
\newcounter{fredcom}
\newtheorem{theorem}{Theorem}
\newtheorem{definition}{Definition}
\newtheorem{fact}{Fact}
\newtheorem{corollary}{Corollary}
\newtheorem{lemma}{Lemma}

\newtheorem{conjecture}{Conjecture}

\newcommand{\node}{\mathrm{node}}

\newcommand{\cl}[1]{\overline{#1}}
\newcommand{\sep}[1]{\overline{\overline{#1}}}

\newcommand{\depth}{\mathrm{depth}}
\newcommand{\timee}{\mathrm{pos}}
\newcommand{\pos}{\timee}
\newcommand{\FCNS}{\mathrm{FCNS}}
\newcommand{\XML}{\mathrm{XML}}
\newcommand{\ann}{\mathrm{ann}}

\begin{document}

\title{
Validating XML Documents in the Streaming Model\\
 with External Memory
\compress{\titlenote{Supported by the French ANR SeSur and Defis programs
under contracts ANR-07-SESU-013 (VERAP project)
and ANR-08-EMER-012 (QRAC project). Christian Konrad is supported by a Fondation CFM-JP Aguilar grant.}}
}
\numberofauthors{2} 
\author{
\alignauthor
Christian Konrad\\
\affaddr{LIAFA, Univ. Paris Diderot; Paris, France;}\\
\affaddr{and Univ. Paris-Sud; Orsay, France.}\\
\email{konrad@lri.fr}
\alignauthor
Fr\'ed\'eric Magniez\\
\affaddr{LIAFA, Univ. Paris Diderot, CNRS; Paris, France.}\\
\email{frederic.magniez@univ-paris-diderot.fr}
}

\maketitle

\begin{abstract} 
We study the problem of validating XML documents of size $N$ against general DTDs in the context of streaming algorithms. The starting point of this work is a well-known space lower bound. There are XML documents and DTDs for which $p$-pass streaming algorithms require $\Omega(N/p)$ space.

We show that when allowing access to external memory, there is a deterministic streaming algorithm that solves this problem with memory space $\Order(\log^2 N)$, a constant number of auxiliary read/write streams, and $\Order(\log N)$ total number of passes on the \acs{XML} document and auxiliary streams. 

An important intermediate step of this algorithm is the computation of the First-Child-Next-Sibling (FCNS) encoding of the initial XML document in a streaming fashion. We study this problem independently, and we also provide memory efficient streaming algorithms for decoding an XML document given in its FCNS encoding.

Furthermore, validating XML documents encoding binary trees in the usual streaming model without external memory can be done with sublinear memory. There is a one-pass algorithm using $\Order(\sqrt{N \log N})$ space, and a bidirectional two-pass algorithm using $\Order( \log^2 N )$ space performing this task.

\end{abstract}


\section{Introduction}
The area of streaming algorithms has experienced tremendous growth over the last decade in many applications.
Streaming algorithms sequentially scan the whole input piece by piece in one pass, or in a small number of passes (i.e., they do not have
random access to the input), while using sublinear memory space, ideally polylogarithmic in the size of the input.
The design of streaming algorithms is motivated by the explosion in the size of the data that algorithms are called upon to process in
everyday real-time applications. Examples of such applications occur in bioinformatics for
genome decoding, in Web databases for the search of documents, or in network monitoring.  
The analysis of Internet traffic \cite{ams99}, in which traffic logs are queried, was one of the first applications of this kind of algorithm.

There are various extensions of this basic streaming model.
One of them gives the streaming algorithm access to an external memory consisting of several read/write streams~\cite{gks05,gs05,ghs06}.
Then the streaming algorithm is also relaxed to perform multiple passes in any direction over the input stream and the auxiliary streams.
In most of the applications, the number of auxiliary streams is constant and the total number of passes is logarithmic in the input size.

Verifying properties or evaluating queries of massive databases is an active and challenging topic.
For relational algebra queries against relational databases, the situation is quite clear.
There are bidirectional $\Order(\log N)$-pass deterministic streaming algorithms with constant memory space and a constant number of auxiliary streams~\cite{ghs09}. Moreover, the logarithmic number of passes is a necessary condition in order to keep the memory space sublinear, even if randomization is allowed. The latter was initially stated for one-sided error~\cite{ghs09} and then extended to two-sided error~\cite{bjr07,bh08}.

In the context of data exchange, especially on the Web, \ac{XML} is emerging as the standard, and is currently drawing much attention in data management research. Only few is known on \ac{XML} query processing when only streaming access is allowed to the \ac{XML} document. For evaluating XQuery and XPath queries against \ac{XML} documents of size $N$, only the lower bound has been extended~\cite{ghs09,bjr07,bh08}, meaning that $\Omega(\log N)$ passes are necessary. For the upper bound, only simple refinements of the direct algorithm are known: no auxiliary stream, one pass and linear memory in the height of the \ac{XML} document, which in the worst case is as large as $N$.

This paper considers the problem of validating \ac{XML} documents against a given \ac{DTD} in a streaming fashion without restrictions on the \ac{DTD}. Prior works on that topic~\cite{sv02,ss07} essentially try to characterize those \ac{DTD}s for which validity can be checked by a finite-state automaton, that is a one-pass deterministic streaming algorithm with constant memory. Concerning arbitrary \ac{DTD}s, two approaches have been considered in~\cite{sv02}. The first one leads to an algorithm with memory space linear in the height of the \ac{XML} document~\cite{sv02}.
The second one consists in constructing 
a refined \ac{DTD} of at most quadratic size, which defines a similar family of tree documents as the original one, and against which validation can be done with constant space. Nonetheless, for an existing document and \ac{DTD}, the later requires that both, documents and DTD, are converted before validation.

One of the obstacles prior works had to cope with was to check well-formedness of \ac{XML} documents, that is every opening tag matches its same-level closing tag. Due to the past work of~\cite{mmn10}, we can now perform such a verification with a constant-pass randomized streaming algorithm with sublinear memory space and no auxiliary streams. In one-pass the memory space is $\Order(\sqrt{N\log N})$, and collapses to $\Order(\log^2 N)$ with an additional pass in reverse direction.

The starting point of this work is the fact that checking validity is hard without auxiliary streams. There are DTDs defining ternary XML documents for which any $p$-pass bidirectional randomized streaming algorithm requires $\Omega(N/p)$ space. 
This lower bound comes from encoding a well known communication complexity problem, Set-Disjointness, as an XML validity problem. This lower bound  should be well-known, however we are not aware of a complete proof in the literature. In~\cite{Grohe:2007}, a similar approach using ternary trees with a reduction to Set-Disjointness is used for proving lower bounds for queries. For the sake of completeness we provide a proof in Appendix~\ref{sec:linear_space_lower_bound}. 

An XML document is valid against a DTD if for each node, the sequence of the labels of their children fulfills a regular expression defined in the DTD. For the case of XML documents encoding binary trees, we present in Section~\ref{section:binary-trees} two 
deterministic streaming algorithms for checking validity with sublinear space. As a consequence,
the presence of nodes of degree at least $3$ is indeed a necessary condition for the linear space lower bound for general documents.
We first show how to design a one-pass algorithm with space $\Order(\sqrt{N \log N} )$ (\textbf{Theorem~\ref{thm:bin-one-pass}}).
We conjecture that there is a $\Omega(\sqrt{N}/p)$ lower bound for $p-$pass algorithms which would render this algorithm tight up to a logarithmic factor, see Appendix~\ref{sec:root_space_lower_bound}.
With a second pass in reverse direction the memory collapses to $\Order(\log^2 N)$ (\textbf{Theorem~\ref{theorem:binary-two-pass}}). These two algorithms make use of the simple, but fundamental fact that in one pass over an XML document each node is seen twice in form of its opening and closing tag. Hence, it is not necessary to remember all opening tags in the stream since there is a second chance to get the same information by their closing tags. Our algorithms exploit this observation. 

Then, in Section~\ref{section:general-trees} we present 
our main result. \textbf{Corollary~\ref{corollary:main-result}} states that validity of any XML document against any \ac{DTD} can be checked in the streaming model with external memory with poly-logarithmic space, a constant number of auxiliary streams, and $\Order(\log N)$ passes over these streams. 
Validity of a node depends on its children, hence is crucial to have easy access to the sequence of children of any node. The fundamental idea to establish this is, firstly, to compute the \ac{FCNS} encoding, an encoding as a $2$-ranked tree of the XML document. In this encoding, the sequence of closing tags of the children of a node are consecutive. The computation of this encoding is the hard part of the validation process, and the resource requirements of our validation algorithm stem from this operation (\textbf{Theorem~\ref{theorem:encoding}}). 
Since the \ac{FCNS} encoding can be seen as a reordering of the tags of the original document, our strategy is to see this problem as a sorting problem with a particular comparison function. Merge sort can be implemented as a streaming algorithm, and we make use of it customized by an adapted merge function. 
The same idea can be used for decoding with similar complexity (\textbf{Theorem~\ref{thm:decoding_log_pass}}).

Then, based on the \ac{FCNS} encoding, verification can be done either in one pass and $\Order(\sqrt{N \log N})$ space (\textbf{Theorem~\ref{thm:rabin-one-pass}}), or in two bidirectional passes and $\Order(\log^2 N)$ space (\textbf{Theorem~\ref{thm:rabin-two-pass}}). Concerning \ac{FCNS} encoding and decoding, we show a linear space lower bound for one-pass algorithms. For decoding, we present a $\Order(\sqrt{N \log N})$ algorithm (\textbf{Theorem~\ref{thm:decoding_two_pass}}) that performs one pass over the input, but two passes over the output. 
We conjecture for encoding that memory space remains $\Omega(N)$ after any constant number of passes, 
which would show that decoding is easier than encoding. 

This suggests a systematic use of the \ac{FCNS} encoding for large documents since validity can be checked easily without auxiliary streams and in sublinear space. For user interactions, the original document can be obtained by the sublinear space $3$-pass algorithm. 
The applicability of this idea is left as an open question.

\section{Preliminaries}
\label{sec:preliminaries}
From now $\Sigma$ is a finite alphabet. The $k$-th letter of $X\in\Sigma^N$ is denoted by $X[k]$, for $1\leq k\leq N$,
and the consecutive letters of $X$ between positions $i$ and $j$ by $X[i,j]$.
A {\em subsequence} of $X$ is any string $X[i_1]X[i_2]\ldots X[i_k]$, where $1\leq i_1<i_2<...<i_k\leq N$.

\subsection{Streaming model}
In streaming algorithms, a {\em pass\/} over input $X\in\Sigma^N$ means that $X$ is given as {\em input stream} $X[1],X[2],\ldots,X[N]$, which arrives sequentially, i.e., letter by letter in this order.
Streaming algorithms have access to random access memory space, and, as the case may be, to read-write external memory as in~\cite{ghs09,dfr06}.
See also the review in~\cite{gks05}.
We assume that any letter of $\Sigma$ fits into one cell of internal/external memory.
The external memory is a collection of auxiliary streams, that we see as {\em read/write streams} with, again, sequential access.
When needed, we augment the alphabet of auxiliary streams from $\Sigma$ by $k$-tuples of elements in $\Sigma \cup \left[0, 2N \right]$, for some fixed constant $k$, which therefore fit in one cell of auxiliary streams. 

At the beginning of each pass on a read/write stream, the algorithm decides whether it performs a read or write pass.
The input stream is read-only. On a writing pass, the algorithm can either write a letter, and then move to the next cell, or move directly to the next cell.
For the case of bidirectional streaming algorithms, opposed to unidirectional streaming algorithm where each pass is in the same order, the algorithm can decide the direction of the sequential access.

For simplicity, we assume throughout this article that the length of the input is known in advance by the algorithm. 
Nonetheless, all our algorithms can be adapted to the case in which the length is unknown until the end of a pass.  
See~\cite{MuthuBook} for an introduction to streaming algorithms.
\begin{definition}[Streaming algorithm]
A $p(N)$-pass 
{\em streaming algorithm} \textbf{A} with $s(N)$ space, $k(N)$ auxiliary streams, $t(N)$ processing time per letter is an 
algorithm such that for every input stream $X\in\Sigma^N$:
\begin{enumerate}
\item \textbf{A} has access to $k(N)$ auxiliary read/write streams,
\item  \textbf{A} performs in total at most $p(N)$ passes on $X$ and auxiliary streams,
\item \textbf{A} maintains a memory space of size $s(N)$ letters of $\Sigma$ and bits while reading $X$ and auxiliary streams,
\item \textbf{A} does not exceed a running time of $t(N)$ between two write or read operations.
\end{enumerate}
We say that \textbf{A} is {\em bidirectional\/} if it performs at least one pass in each direction. 
Otherwise \textbf{A} is implicitly unidirectional.
\end{definition}
We do not mention the number of auxiliary streams when there are none ($k(N)=0$). 
Furthermore, we assume that operations on numbers $N \in [0, 2N]$ can be done in constant time.
\subsection{XML documents}

We consider finite unranked ordered labeled trees $t$, where each tree node is labeled by some label in $\Sigma$,
and its root has a distinguished label $r$. 
Moreover, the children of every non-leaf node are ordered. From now, we omit the terms ordered and labeled.
Then $k$-ranked trees are a special case where each node has at most $k$ children.
Binary trees are a special type of $2$-ranked trees, where each node is either a leaf or has exactly $2$ children.

For each label $a\in\Sigma$, we associate its corresponding {\em opening tag} $a$ and {\em closing tag} $\cl{a}$,
standing for $\langle a\rangle$ and $\langle /a\rangle$ in the usual XML notations.
An {\em XML sequence} is a sequence over the alphabet $\Sigma'=\{ a, \cl{a} : a \in\Sigma\}$.
The {\em XML sequence of a tree $t$} is the sequence of \textit{opening tags} and \textit{closing tags} in the order of a depth first traversal of $t$ (Figure~\ref{fig:xml-example}): when at step $i$ we visit a node with label $a$ top-down (respectively bottom-up), we let $X[i]=a$ (respectively  $X[i]=\cl{a}$).
Hence $X$ is a word over $\Sigma'=\{ a, \cl{a} : a \in\Sigma\}$ of size twice the number of nodes of $t$. 
The XML file describing $t$ is unique, and we denote it as $\XML(t)$.


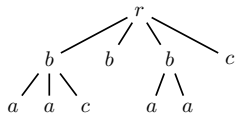
\begin{figure}[h]
\centerline{\scalebox{0.8}{
\begin{tikzpicture}[node distance = 0.5 cm]
\tikzstyle{VertexStyle} = [shape = rectangle, color = white, fill = white, text = black, draw]
\Vertex[x=4, y=4, L=$r$]{r} 
\Vertex[x=2.5, y=3.2,L=$b$]{b1}
\Vertex[x=3.5, y=3.2, L=$b$]{b2}
\Vertex[x=4.5, y=3.2, L=$b$]{b3}
\Vertex[x=5.5, y=3.2, L=$c$]{c1}
\Vertex[x=1.9, y=2.4, L=$a$]{a1}
\Vertex[x=2.5, y=2.4, L=$a$]{a2}
\Vertex[x=3.1, y=2.4, L=$c$]{c2}
\Vertex[x=4.2, y=2.4, L=$a$]{a3}
\Vertex[x=4.8, y=2.4, L=$a$]{a4}  
\Edge(r)(b1)
\Edge(r)(b2)
\Edge(r)(b3)
\Edge(r)(c1)
\Edge(b1)(a1)
\Edge(b1)(a2)
\Edge(b1)(c2)
\Edge(b3)(a3)   
\Edge(b3)(a4)
\end{tikzpicture}
}}
\caption{Let $\Sigma = \{a, b, c \}$, and let $t$ be the tree as above. Then $\XML(t)=rba\cl{a}a\cl{a}c\cl{c}\cl{b}b\cl{b}ba\cl{a}a\cl{a}\cl{b}c\cl{c}\cl{r}$.\label{fig:xml-example}} 
\end{figure}

We assume that the input \ac{XML} sequences $X$ are {\em well-formed}, namely $X=\XML(t)$, for some tree $t$. The past work of~\cite{mmn10} legitimates this assumption since checking well-formedness is at least as easy as any of our algorithms for checking validity. Hence we could run an algorithm for well-formedness in parallel without increasing the resource requirements. Note, that randomness is necessary for checking well-formedness with sublinear space, whereas we will show that randomness is useless for validation.

Let us introduce more useful notations. Since the length of a well-formed \ac{XML} sequence is known in advance, we will denote it by $2N$ instead of $N$.
Each opening tag $X[i]$ and matching closing tag $X[j]$ in $X=\XML(t)$ corresponds to a unique tree node $v$ of $t$.
We sometimes denote $v$ either by $X[i]$ or $X[j]$. 
We also write (ambiguously) $v$ for its corresponding opening tag, and $\cl{v}$ for its corresponding closing tag.
Then, the {\em position} of $v$ in $X$ is $\pos(v)=i$. Similarly, $\pos(\cl{v})=j$.


We consider \ac{XML} validity against some \ac{DTD}. A \ac{DTD} is a mapping $D$ from $\Sigma$ to regular expressions over $\Sigma$.
Let $t$ be a tree. Then a node $v\in t$ with label $v$ and children $v_1,v_2,\ldots, v_k$ with respective labels $v_1,v_2,\ldots,v_k$ is {\em valid against $D$} if $v_1,v_2,\ldots,v_k$ satisfies the regular expression $D(v)$.
In particular, $v$ can be a leaf if and only if the empty word $\eps$ satisfies the regular expression $D(v)$.
Then $t$ is {\em valid against $D$} if all its nodes are valid against $D$.
Throughout the document we assume that \ac{DTD}s are considerably small and our algorithms have full access to them without accounting this to their space requirements.  

\begin{definition}[\pb]
Let $D$ be some \ac{DTD}. The problem $\pb$ consists of deciding whether an input tree $t$ given by its XML sequence $\XML(t)$ on an input stream is valid against $D$.
\end{definition}
We denote by $\pb(2)$ the problem $\pb$ restricted to input XML sequences describing binary trees.

\section{Validity of binary trees}
\label{section:binary-trees}
For simplicity, we only consider binary trees in this section.
A {\em left opening/closing tag} (respectively {\em right opening/closing tag}) of an XML sequence $X$ is a tag whose corresponding node is the first child of its parent (respectively second child).

Our algorithms for binary trees can be extended to $2$-ranked trees. This requires few changes in the one-pass Algorithm~\ref{algo_bin_tree_one_pass} and the two-pass Algorithm~\ref{algo_bin_tree_two_passes} (indeed in the subroutine Algorithm~\ref{algo_bin_tree_subroutine}), that we do not describe here.

We fix now a \ac{DTD} $D$, and assume that, in our algorithms, we have access to a procedure check$(a,b,c)$ that signalizes invalidity and aborts if $bc$ is not valid against the regular expression $D(a)$. Otherwise it returns without any action.


\subsection{One-pass algorithm}
In order to validate an XML document, we ensure validity of all tree nodes. For checking validity of a node $v$ with two children, we have to relate $3$ labels, that is the label $v$ of the node itself, and the labels of the two children nodes $v_1, v_2$. In a \textit{top-down} verification we use the opening tag $v$ of the parent node $v$ for verification, in a \textit{bottom-up} verification we use the closing tag $\cl{v}$ of the parent node $v$. Algorithm~\ref{algo_bin_tree_one_pass} makes use of the fact that there are these two chances to verify a node. It uses a stack onto which it pushes all opening tags in order to perform top-down verifications once the information of the children nodes arrives on the stream. $\cl{v_1}v_2$ forms a substring of the input, hence top-down verification requires only the storage of the opening tag $v$ since the labels of the children arrive in a block. The algorithm's space requirements depend on a parameter $K$ (we optimize by setting $K = \sqrt{N \log N}$). Once the number of opening tags on the stack is about to exceed $K$, we remove the bottom-most opening tag. The corresponding node will then be verified bottom-up. Note that $\cl{v_2}\cl{v}$ forms a substring of the input. Hence for bottom-up verifications it is enough to store the label of the left child $v_1$ on the stack since the label of the right child arrives in form of a closing tag right before the closing tag of the parent node. See Algorithm~\ref{algo_bin_tree_one_pass} for details. 

For the unique identification of closing tags on the stack, we have to store them with their depth in the tree. A stack item corresponding to a closing tag requires hence $\Order(\log N)$ space. Opening tags don't require the storage of their depth (we store a depth of $-1$ which we assume to require only constant space).

\begin{algorithm}[H]\small
\caption{Validity of binary trees in 1-pass \label{algo_bin_tree_one_pass}}
\begin{algorithmic}[1]
\REQUIRE input stream is a well-formed XML document
\STATE $d \gets 0, S \leftarrow$ empty stack
\STATE $K \gets \sqrt{N \log N}$ \label{line:setting_K}
\WHILE{stream not empty} \label{line:beginning_of_while_loop}
\STATE $x \gets $ next tag on stream 
\IF{$x$ is an opening tag $c$}
\STATE \textbf{if} $x$ is a leaf \label{line:leaf_check} \textbf{then} check$(c, \epsilon, \epsilon)$ \label{line:check_leaf} \textbf{end if}
\IF{$S$ has on top $(a, -1), (\cl{b}, d)$} \label{line:top_down_verification}
  \STATE check($a, b, c$); pop $S$ \COMMENT{\textit{Top-down verification}} 
\ENDIF
\IF{$|\{ (a, -1 ) \in S \, | \, a \mbox{ opening } \}| \ge K$} \label{clean-stack}
  \STATE remove bottom-most $(a, -1 )$ in $S$, $a$ opening
\ENDIF
\STATE $d \gets d + 1$
\STATE push $(x, -1 )$ \label{line:push_opening}
\ELSIF{$x$ is a closing tag $\cl{c}$}
\STATE $d \gets d - 1$
\STATE \textbf{if} S has on top $(\cl{a}, d+1), (\cl{b}, d+1)$ \textbf{then} \label{line:bottom_up_verification}
   \STATE \hspace{0.3cm} check ($c, a, b$) \COMMENT{\textit{Bottom-up verification}} 
   \STATE \hspace{0.3cm} pop $S$, pop $S$
\STATE \textbf{else if} $S$ has on top $(\cl{b}, d+1)$ \textbf{then} pop $S$ \label{line:remove_1}
\STATE \textbf{end if}
\STATE \textbf{if} $S$ has on top $(c, -1 )$ \textbf{then} pop $S$ \textbf{end if} \label{line:remove_2}
\STATE push $(x, d)$ \label{line:push_closing}
\ENDIF
\ENDWHILE
\end{algorithmic}
\end{algorithm}

The query in line~\ref{line:leaf_check} can be implemented by a lookahead of $1$ on the stream. The opening tag $x$ corresponds to a leaf only if the subsequent tag in the stream is the corresponding closing tag $\cl{x}$.

Figure~\ref{figure:visualisation_algo_1} visualizes the different cases  with their stack modifications appearing in Algorithm~\ref{algo_bin_tree_one_pass}.

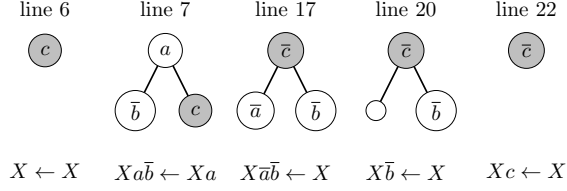
\begin{figure}[h]
\centerline{\scalebox{0.8}{
\begin{tikzpicture}
\tikzstyle{VertexStyle} = [shape = circle, color = black, fill = lightgray, text = black, draw]
\Vertex[x=0.5, y=1, L=$c$]{c0}
\Vertex[x=3, y=0, L=$c$]{c2}
\Vertex[x=4.5, y=1, L=$\cl{c}$]{c1}
\Vertex[x=6.5, y=1, L=$\cl{c}$]{c3}
\Vertex[x=8.5, y=1, L=$\cl{c}$]{c4}
\tikzstyle{VertexStyle} = [shape = circle, color = black, fill = white, text = black, draw]
\Vertex[x=2.5, y=1, L=$a$]{a2}
\Vertex[x=2, y=0, L=$\cl{b}$]{b2}
\Edge(a2)(b2)
\Edge(a2)(c2)
\Vertex[x=4, y=0, L=$\cl{a}$]{a1}
\Vertex[x=5, y=0, L=$\cl{b}$]{b1}
\Edge(a1)(c1)
\Edge(b1)(c1)
\Vertex[x=6, y=0, L= $ $]{a3}
\Vertex[x=7, y=0, L=$\cl{b}$]{b3}
\Edge(c3)(a3)
\Edge(c3)(b3)
\draw(0.5, 1.7) node{line~\ref{line:leaf_check}};
\draw(2.5, 1.7) node{line~\ref{line:top_down_verification}};
\draw(4.5, 1.7) node{line~\ref{line:bottom_up_verification}};
\draw(6.5, 1.7) node{line~\ref{line:remove_1}};
\draw(8.5, 1.7) node{line~\ref{line:remove_2}}; 
\draw(0.5, -1) node{$X \gets X$};
\draw(2.5, -1) node{$X a\cl{b} \gets X a$};
\draw(4.5, -1) node{$X \cl{a}\cl{b} \gets X $};
\draw(6.5, -1) node{$X \cl{b} \gets X$};
\draw(8.5, -1) node{$X c \gets X$};                                
\end{tikzpicture}
}}

\caption{ Visualization of the different conditions in Algorithm~\ref{algo_bin_tree_one_pass} with the applied stack modifications. $X$ represents the bottom part of the stack. Note that Algorithm~\ref{algo_bin_tree_one_pass} pushes the currently treated tag $c$ or $\cl{c}$ on the stack in Line~\ref{line:push_opening} or Line~\ref{line:push_closing}. $c$ or $\cl{c}$ corresponds to the highlighted node. \label{figure:visualisation_algo_1}} 
\end{figure}

Fact~\ref{fact:stack} (not proved here) and Lemma~\ref{lemma:stack} concern the structure of the stack $S$ used in Algorithm~\ref{algo_bin_tree_one_pass}. 

\begin{fact}
\label{fact:stack}
Let $S = (x_1, d_1), \dots (x_k, d_k)$ be the stack at the beginning of the while loop in line~\ref{line:beginning_of_while_loop}. Then:
\begin{enumerate}
\item $\pos(x_1) < \pos(x_2) \dots < \pos(x_k)$, \label{item:fact_stack_1}
\item $\depth(x_1) \le \depth(x_2) \dots \le \depth(x_k) \le d$. Moreover, if $\depth(x_i) = \depth(x_{i+1})$ then $x_i$ is the left sibling of $x_{i+1}$,
\item The sequence $x_1 \dots x_k$ satisfies the regular expression $\cl{a}^*b^* ( \epsilon \, | \, \cl{c} \, | \, \cl{d}\cl{e})$, where $\cl{a}^*$ are left closing tags, $b^*$ are opening tags, $\cl{c}$ is a closing tag, $d$ is a left closing tag, and $e$ is a right closing tag. \label{item:fact_stack_3}
\item A left closing tag $\cl{a}$ is only removed from $S$ upon verification of its parent node. \label{item:fact_stack_4}
\end{enumerate}
\end{fact}

\begin{lemma}
\label{lemma:stack}
Let $S = (x_1, d_1), \dots (x_k, d_k)$ be the stack at the beginning of the while loop in line~\ref{line:beginning_of_while_loop}.
Let $(\cl{c_i}, d_i), (\cl{c_{i+1}}, d_{i+1})$ be two consecutive left closing tags in $S$ such that $(\cl{c_{i+1}}, d_{i+1})$ is not the topmost one. Then $\pos({\cl{c_{i+1}}}) \ge \pos({\cl{c_i}}) + 2K$.

\end{lemma}

\begin{proof}
Denote by $X=X[1]X[2] \dots X[2N]$ the input stream. Since $\cl{c_{i+1}}$ is not the topmost left closing tag in $S$, the algorithm has already processed the right sibling opening tag $X[\pos(\cl{c_{i+1}})+1]$ of $\cl{c_{i+1}}$. By Item~\ref{item:fact_stack_4} of Fact~\ref{fact:stack}, no verification has been done of the parent of $\cl{c_{i+1}}$, since $\cl{c_{i+1}}$ is still in $S$. Therefore, the parent's opening tag $X[k]$ of $\cl{c_{i+1}}$ has been deleted from $S$, where $\pos(\cl{c_{i}}) < k < \pos(\cl{c_{i+1}})$. This can only happen if at least $K$ opening tags have been pushed on $S$ between $X[k]$ and $\cl{c_{i+1}}$. Since these $K$ opening tags must have been closed between $X[k]$ and $\cl{c_{i+1}}$ we obtain $\pos({\cl{c_{i+1}}}) \ge \pos({\cl{c_i}}) + 2K$.
\end{proof}

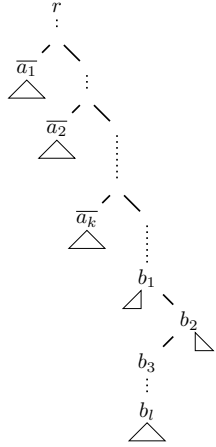
\begin{figure}[h]
\centerline{\scalebox{0.8}{
\begin{tikzpicture}
\tikzstyle{VertexStyle} = [shape = rectangle, color = white, fill = white, text = black, draw]
\Vertex[x=4, y=10, L=$r$]{r} 
\Vertex[x=4, y=9.5, L=$ $]{p1}
\Vertex[x=3.5, y=9, L=$\cl{a_1}$]{a1}
\Vertex[x=4.5, y=9, L=$ $]{b1}
\Edge[style=dotted](r)(p1)
\Edge(p1)(a1)
\Edge(p1)(b1)
\draw (3.5, 8.8) -- (3.2, 8.5) -- (3.8, 8.5) -- cycle;
\Vertex[x=4.5, y=8.5, L=$ $]{p2}
\Edge[style=dotted](b1)(p2)
\Vertex[x=4, y=8, L=$\cl{a_2}$]{a2}
\Vertex[x=5, y=8, L=$ $]{b2}
\Edge(p2)(a2)
\Edge(p2)(b2)
\draw (4, 7.8) -- (3.7, 7.5) -- (4.3, 7.5) -- cycle;
\Vertex[x=5, y=7, L=$ $]{pk}
\Edge[style=dotted](b2)(pk)
\Vertex[x=4.5, y=6.5, L=$\cl{a_k}$]{ak}
\Vertex[x=5.5, y=6.5, L=$ $]{bk}
\Edge(pk)(ak)
\Edge(pk)(bk)
\draw (4.5, 6.3) -- (4.2, 6) -- (4.8, 6) -- cycle;
\Vertex[x=5.5, y=5.5, L=$b_1$]{bb1}
\draw (5.4, 5.3) -- (5.1, 5) -- (5.4, 5) -- cycle;
\Vertex[x=6.2, y=4.8, L=$b_2$]{bb2}
\draw (6.3, 4.6) -- (6.3, 4.3) -- (6.6, 4.3) -- cycle;
\Vertex[x=5.5, y=4.1, L=$b_3$]{bb3}
\Edge[style=dotted](bk)(bb1)
\Edge(bb1)(bb2)
\Edge(bb2)(bb3)
\Vertex[x=5.5, y=3.3, L=$b_l$]{bbk}
\Edge[style=dotted](bb3)(bbk)
\draw (5.5, 3.1) -- (5.2, 2.8) -- (5.8, 2.8) -- cycle;
\end{tikzpicture}
}}
\caption{Visualization of the structure of the stack used in Algorithm~\ref{algo_bin_tree_one_pass}. The stack fulfills the regular expression $\cl{a}^*b^*(\epsilon \, | \, \cl{c} \, | \, \cl{d} \cl{e})$, compare Item~\ref{item:fact_stack_3} of Fact~\ref{fact:stack}. The $(\cl{a_i})_{i=1 \dots k}$ are closing tags whose parents' nodes were not verified top-down. For $j > i$, $a_j$ is connected to $a_i$ by the right sibling of $a_i$. The $(b_i)_{i=1 \dots l}$ form a sequence of opening tags such that $b_i$ is the parent node of $b_{i+1}$. On top of the stack might be one or two closing tags depending on the current state of the verification process. \label{figure:visualisation_algo_11}} 
\end{figure}

Fact~\ref{fact:stack} and Lemma~\ref{lemma:stack} provide more insight in the stack structure and are used in the proof of Theorem~\ref{thm:bin-one-pass}. Item~\ref{item:fact_stack_3} of Fact~\ref{fact:stack} states that the stack basically consists of a sequence of left closing tags which are the left children that are needed for bottom-up verifications of nodes that could not be verified top-down. This sequence is followed by a sequence of opening tags for which we still aim a top-down verification. The proof of Lemma~\ref{lemma:stack} explains the fact that the two sequences are strictly separated: a left-closing tag $\cl{v_1}$ only remains on the stack if at the moment of insertion there are no opening tags on the stack.

\begin{theorem}\label{thm:bin-one-pass}
Algorithm~\ref{algo_bin_tree_one_pass} is a one-pass streaming algorithm for $\pb(2)$
with space $\Order(\sqrt{N \log N})$ and $\Order(1)$ processing time per letter.
\end{theorem}
\begin{proof}
To prove correctness, we have to ensure validity of all nodes. Leaves are
correctly validated upon arrival of its opening tag in line~\ref{line:check_leaf}.
Concerning non-leaf nodes, firstly, note that all closing tags are pushed on $S$ in
line~\ref{line:push_closing}, in particular all closing tags of left children appear
on the stack. The algorithm removes left closing tags only after validation of
its parent node, no matter whether the verification was done top-down or bottom-up, compare Item~\ref{item:fact_stack_4} of Fact~\ref{fact:stack}. 
Emptiness of the stack after the execution of the algorithm
follows from Item~\ref{item:fact_stack_1} of Lemma~\ref{lemma:stack} and implies hence the validation 
of all non-leaf nodes.

For the space bound, Line~\ref{clean-stack} guaranties that the number of opening tags in $S$ is always at most $K$. 
We bound the number of closing tags on the stack by $\frac{N}{K} + 2$. Item~$3$ of Lemma~\ref{lemma:stack} states that the 
stack contains at most one right closing tag. From Item~\ref{item:fact_stack_4} of Lemma~\ref{lemma:stack} we 
deduce that $S$ comprises at most $\frac{N}{K}+1$ left closing tags, since the stream is of length $2N$, 
and the distance in the stream of two consecutive left closing tags that reside on $S$ except the top-most one is at least $2K$. 
A closing tag with depth $(a, d) \in \Sigma' \times [N]$ requires $\Order(\log N)$ space, an opening tag requires only
constant space. Hence the total space requirements are $\Order((\frac{N}{K} + 2) \log N + K)$ which is minimized for $K=\sqrt{N \log N}$.

Concerning the processing time per letter, the algorithm only performs a constant 
number of local stack operations in one iteration of the while loop. 
\end{proof}

\textbf{Remark} Algorithm~\ref{algo_bin_tree_one_pass} can be turned into an algorithm with
space complexity $\Order(\sqrt{D \log D})$, where $D$ is the depth of the XML document. If $D$
is known beforehand, it is enough to set $K = \sqrt{D \log D}$ in line~\ref{line:setting_K}. If $D$
is not known in advance, we make use of an auxiliary variable $D'$ storing a guess for the 
document depth. Initially we set $D' = C$, $C > 0$ some constant, we set $K = \sqrt{D' \log D'}$, and we
run Algorithm~\ref{algo_bin_tree_one_pass}. Each time $d$ exceeds $D'$, we double $D'$, and 
we update $K$ accordingly.

This guarantees that the number of opening tags on the stack is limited by $\Order(\sqrt{D \log D})$. 
Since we started with a too small guess for the document depth, we may have removed opening tags 
that would have remained on the stack if we had chosen the
depth correctly. This leads to further bottom-up verifications, but no more than $\Order(\sqrt{D/ \log D})$
guaranteeing $\Order(\sqrt{D \log D})$ space.

\subsection{Two-pass algorithm}

%
%
%
%
The bidirectional two-pass Algorithm~\ref{algo_bin_tree_two_passes} uses a subroutine that checks in one-pass validity of all nodes whose left subtree is at least as large as its right subtree. Feeding into this subroutine the XML document read in reverse direction and interpreting opening tags as closing tags and vice versa, it checks validity of all nodes whose right subtree is at least as large as its left subtree. In this way all tree nodes get verified.

The subroutine performs only checks in a bottom-up fashion, that is, the verification of a node $v$ with children $c_1, c_2$ makes use of the tags $\cl{c_1}$ and $c_2$ (which are adjacent in the XML document and hence easy to recognize) and the closing tag of $\cl{v}$. When $\cl{c_1}, c_2$ appears in the stream, a 4-tuple consisting of $\cl{c_1}, c_2, \depth(c_1)$ and $\pos(\cl{c_1})$ gets pushed on the stack. Upon arrival of $\cl{v}$, $\depth(c_1)$ is needed to identify $c_1, c_2$ as the children of $v$.  $\pos(\cl{c_1})$ is needed for cleaning the stack: with the help of the $\pos$ values of the stack items, we identify stack items whose parents' nodes have larger right subtrees than left subtrees, and these stack items get removed from the stack. In so doing, we guarantee that the stack size does not exceed $\log(N)$ elements which is an exponential improvement over the one-pass algorithm.

Note that the reverse pass can be done independently of the first one, namely eventually in parallel.

\begin{algorithm}[h]\small
\caption{Two-pass algorithm validating binary trees\label{algo_bin_tree_two_passes}}
\begin{algorithmic}
\STATE run \textbf{Algorithm~\ref{algo_bin_tree_subroutine}} reading the stream from left to right
\STATE run \textbf{Algorithm~\ref{algo_bin_tree_subroutine}} reading the stream from right to left, where opening tags are interpreted as closing tags, and vice versa.
\end{algorithmic}
\end{algorithm}

\begin{algorithm}[H]\small
\caption{Validating nodes with size(left subtree) $\geq$ size(right subtree)\label{algo_bin_tree_subroutine}}
\begin{algorithmic}[1]
\STATE $l \gets 0$; $n \gets 0$; $S \leftarrow$ empty stack
\WHILE{stream not empty}
\STATE $x \gets $ next tag on stream (and move stream to next tag)
\STATE $y\gets$ next tag on stream, without consuming it yet
\STATE $n \gets n + 1$
\IF{$x$ is an opening tag $c$}
 \STATE $l \gets l + 1$
 \STATE \textbf{if} $y = \cl{c}$ \textbf{then}  check($c, \epsilon, \epsilon$) \textbf{end if} \label{line:leaf_check_algo2}  
\ELSE[$x$ is a closing tag $\cl{c}$]
 \STATE $l \gets l-1$
 \IF{$S$ has on top $(\cdot, \cdot, l+1, \cdot)$ \label{line:check_item_algo2}}
   \STATE $(\cl{a},b,\cdot,\cdot)\gets$ pop from $S$; check($c, a, b$) 
 \ENDIF   
 \IF{$y$ is an opening tag $d$ \label{line:push_item_algo2}}     
   \STATE push ($\cl{c}$, $d$, $l$, $n$) to $S$ 
 \ENDIF  
\ENDIF
\WHILE{there is $s_1=(\cdot, \cdot, \cdot, n_1)$ just below $s_2=(\cdot, \cdot, \cdot, n_2)$ in $S$ with $n-n_2 > n_2 - n_1$}\label{line:clean_stack}
\STATE suppress $s_2$ from $S$ 
\ENDWHILE \label{line:end_while_loop}
\ENDWHILE
\end{algorithmic}
\end{algorithm}

Figure~\ref{figure:visualisation_algo_2} visualizes the different cases in Algorithm~\ref{algo_bin_tree_subroutine}.

\begin{figure}[h]
\centerline{\scalebox{0.8}{
\begin{tikzpicture}
\tikzstyle{VertexStyle} = [shape = circle, color = black, fill = lightgray, text = black, draw]
\Vertex[x=0.5, y=1, L=$c$]{c0}
\Vertex[x=2.5, y=1, L=$\cl{c}$]{c1}
\Vertex[x=4, y=0, L=$\cl{c}$]{c2}
\tikzstyle{VertexStyle} = [shape = circle, color = black, fill = white, text = black, draw]
\Vertex[x=2, y=0, L=$\cl{a}$]{a1}
\Vertex[x=3, y=0, L=$b$]{b1}
\Edge(c1)(a1)
\Edge(c1)(b1)
\Vertex[x=4.5, y=1, L=$ $]{a2}
\Vertex[x=5, y=0, L=$d$]{d2}
\Edge(a2)(c2)
\Edge(a2)(d2)
\draw(0.5, 1.7) node{line~\ref{line:leaf_check_algo2}};
\draw(2.5, 1.7) node{line~\ref{line:check_item_algo2}};
\draw(4.5, 1.7) node{line~\ref{line:push_item_algo2}};
\end{tikzpicture}
}}

\caption{ Visualization of the different conditions in Algorithm~\ref{algo_bin_tree_subroutine}. The incoming tag $x$ corresponds to the highlighted node. \label{figure:visualisation_algo_2}} 
\end{figure}
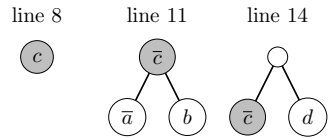

We highlight some properties concerning the stack used in Algorithm~\ref{algo_bin_tree_subroutine}.

\begin{fact}
\label{fact_stack_order_two_passes}
$S$ in Algorithm~\ref{algo_bin_tree_subroutine} satisfies the following:
\begin{enumerate}
\item If $(\cl{a_1}, b_1, \depth(\cl{a_1}), \pos(a_1))$ is below $(\cl{a_2}, b_2, \depth(\cl{a_2}), \pos(a_2))$ in $S$, then
$\mbox{$\pos(\cl{a_1}) <\pos(\cl{a_2})$}$, $\mbox{$\depth(\cl{a_1}) <\depth(\cl{a_2})$}$, and $a_2, b_2$ are in the subtree of $b_1$. \label{fact:two_stack_elements}
\item Consider $l$ at the end of the while loop in line~\ref{line:end_while_loop}. Then there are no stack elements $(\cdot, \cdot, l', \cdot)$ with $l' > l$. \label{fact:no_deeper_elements}
\end{enumerate}
\end{fact}

Figure~\ref{figure:stack_algo_2} illustrates the relationship between two consecutive stack elements as discussed in Item~\ref{fact:two_stack_elements} of Fact~\ref{fact_stack_order_two_passes}.

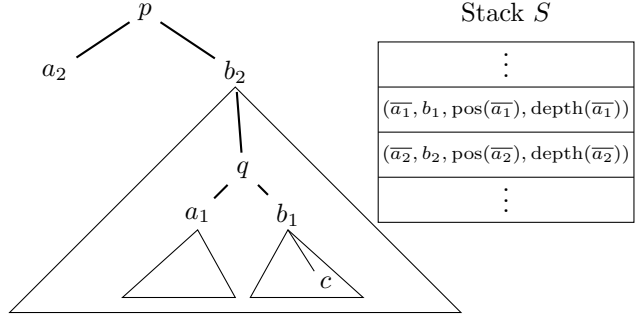
\begin{figure}[h]
\centerline{\scalebox{1}{
\begin{tikzpicture}[node distance = 0.5 cm]
\tikzstyle{VertexStyle} = [shape = rectangle, color = white, fill = white, text = black, draw]
\Vertex[x=4, y=4, L=$p$]{p} 
\Vertex[x=2.8, y=3.2,L=$a_2$]{v1}
\Vertex[x=5.2, y=3.2, L=$b_2$]{v2}
\Edge(p)(v1)
\Edge(p)(v2)
\Vertex[x=5.3, y=1.9, L=$q$]{q} 
\Vertex[x=4.7, y=1.3, L=$a_1$]{u1} 
\Vertex[x=5.9, y=1.3, L=$b_1$]{u2} 
\Edge[color=black](q)(u1)
\Edge(q)(u2)
\Vertex[x=6.4, y=0.4, L=$c$]{c} 
\draw (5.2, 3.0) -- (2.2, 0) -- (8.2, 0) -- cycle;
\draw[color=black] (4.7, 1.1) -- (3.7, 0.2) -- (5.2, 0.2) -- cycle;
\draw (5.9, 1.1) -- (5.4, 0.2) -- (6.9, 0.2) -- cycle;
\Edge[color=black](v2)(q)
\draw[color=black] (5.9, 1.1) -- (6.25, 0.55);
%
%
\draw (7.1, 3.6) -- (10.5, 3.6) -- (10.5, 1.2) -- (7.1, 1.2) -- cycle;
\draw (7.1, 3.0) -- (10.5, 3.0);
\draw (7.1, 2.4) -- (10.5, 2.4);
\draw (7.1, 1.8) -- (10.5, 1.8);
\draw(8.8, 4) node{Stack $S$};
\draw(8.8, 3.4) node{$\vdots$};
\draw(8.8, 2.7) node{\scriptsize $(\cl{a_1}, b_1, \pos(\cl{a_1}), \depth(\cl{a_1}) )$};
\draw(8.8, 2.1) node{\scriptsize $(\cl{a_2}, b_2, \pos(\cl{a_2}), \depth(\cl{a_2}) )$};
\draw(8.8, 1.6) node{$\vdots$};
\end{tikzpicture} 
}}
\caption{$c$ is the current element under consideration in Algorithm~\ref{algo_bin_tree_subroutine}. $a_1, b_1$ is in the subtree of $b_2$, compare Item~\ref{fact:two_stack_elements} of Fact~\ref{fact_stack_order_two_passes}. \label{figure:stack_algo_2}}. 
\end{figure}

\begin{lemma}
\label{lemma:verification_subtree}
Algorithm~\ref{algo_bin_tree_subroutine} verifies all nodes $v$ whose left subtree is at 
least as large as its right subtree.
\end{lemma}

\begin{proof}
Let $q$ be such a node. Let $a_1, b_1$ be the children of $q$. Then it holds 
that 
\begin{equation}
\label{eqn:subtree_condition}
\pos(\cl{a_1}) - \pos(a_1) \ge \pos(\cl{b_1}) - \pos(b_1), 
\end{equation}
since the size of the left subtree of $q$ is at least as large as the size of the right subtree.

Upon arrival of $\cl{a_1}$  Algorithm~\ref{algo_bin_tree_subroutine} pushes the 
4-tuple $\mbox{$t = (\cl{a_1}, b_1, \pos(\cl{a_1}), \depth(a_1))$}$ onto the stack $S$.
We have to show that $t$ remains on the stack until the arrival of $\cl{q}$. More precisely, 
we have to show that the condition in line~\ref{line:clean_stack} is never satisfied for $s_2 = t$. 
Since the algorithm never deletes the bottom-most stack item, we consider the case where there is a stack item 
$(\cl{a_2}, b_2, \pos(\cl{a_2}), \depth(a_2))$ just below $t$. Item~\ref{fact:two_stack_elements} of 
Fact~\ref{fact_stack_order_two_passes} tells us that $a_1, b_1$ are in the subtree of $b_2$. 
Let $c$ be the current tag under consideration such that $\pos(b_1) < \pos(c) < \pos(\cl{q})$. 
The situation is visualized in Figure~\ref{figure:stack_algo_2}.

According to the condition of line~\ref{line:clean_stack}, $t$ gets removed from the stack if 
\begin{equation}
\label{equ:deletion_condition}
\pos(c) - \pos(\cl{a_1}) > \pos(\cl{a_1}) - \pos(\cl{a_2}) .
\end{equation}
Note that the left side of Inequality~\ref{equ:deletion_condition} is a lower
bound on the size of the right subtree of $q$. Furthermore, the right side of
Inequality~\ref{equ:deletion_condition} upper bounds the size of the left 
subtree of $q$.

Using $\pos(c) - \pos(\cl{a_1}) \le \pos(\cl{b_1}) - \pos(b_1) + 1$ and 
$\pos(\cl{a_1}) - \pos(\cl{a_2}) > \pos(\cl{a_1}) - \pos(a_1)$, Inequality~\ref{equ:deletion_condition}
contradicts Inequality~\ref{eqn:subtree_condition} which shows that $t$ remains
on the stack until the arrival of $\cl{q}$. Item~\ref{fact:no_deeper_elements} of Fact~\ref{fact_stack_order_two_passes}
guarantees that there is no other stack element on top of $t$ upon arrival of $\cl{q}$. 
This guarantees the verification of node $q$ and proves the lemma.
\end{proof}

\begin{theorem}
\label{theorem:binary-two-pass}
Algorithm~\ref{algo_bin_tree_two_passes} is a bidirectional two-pass streaming algorithm for $\pb(2)$
with space $\Order(\log^2{N})$ and $\Order(\log N)$ processing time per letter.
\end{theorem}
\begin{proof}

To prove correctness of Algorithm~\ref{algo_bin_tree_two_passes}, we ensure
that all nodes get verified. By Lemma~\ref{lemma:verification_subtree}, in the first pass,
all nodes with a left subtree being at least as large as its right subtree get verified. 
The second pass ensures then verification of nodes with a right subtree that is at least
as large as its left subtree.

Next, we prove by contradiction that for any current value of variable $n$ in Algorithm~\ref{algo_bin_tree_subroutine},
the stack contains at most $\log(n)$ elements. Assume that there is a stack configuration of size $t\ge\log(n) + 1$. 
Let $(n_1, n_2 \dots, n_t)$ be the sequence of the fourth parameters of the stack elements. Since these elements are 
not yet removed, due to line~\ref{line:clean_stack} of Algorithm~\ref{algo_bin_tree_subroutine}, it holds that 
$\mbox{$n - n_i \le n_i - n_{i-1}$}$, or equivalently $\mbox{$n_i \ge 1/2(n+n_{i-1})$}$, for all $\mbox{$1 < i \le t$}$.
Since $n_1 \ge 1$, we obtain that $n_i \ge \frac{2^i - 1}{2^i} n + \frac{1}{2^i}$, and, in particular, 
$n_{t-1} \ge (n-1) + \frac{1}{n}$. Since all $n_i$ are integers, it holds that $n_{t-1} \ge n$. Furthermore, since $n_{t} > n_{t-1}$, we obtain $n_{\log n + 1} \ge n + 1$ which is a contradiction, since the element at position $n+1$ has not yet been seen.

Since $n\leq 2N$ and the size of a stack element is in $\Order(\log n)$,
Algorithm~\ref{algo_bin_tree_subroutine} uses space $\Order(\log^2 N)$. This also implies that
the while-loop at line~\ref{line:clean_stack} of Algorithm~\ref{algo_bin_tree_subroutine} can 
only be iterated $\Order(\log n)$ times during the processing of a tag on the stream. 
The processing time per letter is then $\Order(\log N)$, 
since we assume that operations on the stack run in constant time.
\end{proof}


\section{Validity of general trees}
\label{section:general-trees}
\subsection{Preparation}
The {\em \ac{FCNS} encoding} (see for instance \cite{Neven02automatatheory}) is an encoding of unranked trees as {\em extended} $2$-ranked trees, where we distinguish left child from right child. This is an extension of ordered $2$-ranked trees, since a node may have a left child but no right child, and vice versa. We therefore duplicate the labels $a\in \Sigma$ to $a_L$ and $a_R$, for respectively {\em left} and {\em right} opening/closing tags. The \ac{FCNS} tree is obtained by keeping the same set of tree nodes. The root node of the unranked tree remains the root in the \ac{FCNS} tree, and we annotate it by default left. The left child of any internal node in the \ac{FCNS} tree is the first child of this node in the unranked tree if it exists, otherwise it does not have a left child. The right child of a node in the \ac{FCNS} tree is the next sibling of this node in the unranked tree if it exists, otherwise it does not have a right child.
For a tree $t$, we denote $\FCNS(t)$ the \ac{FCNS} tree, and $\XML(\FCNS(t))$ the XML sequence of the FCNS encoding of $t$. 

Instead of annotating by left/right, another way to uniquely identify a node as left or right is to insert dummy leaves with label $\bot$. For a tree $t$, we denote the binary version without annotations and insertion of $\bot$ leaves by $\FCNS^{\bot}$. The two representations can be easily transformed into each other. In this section, we compute the FCNS encoding with annotations. In the next section, we present algorithms for the validation of the encoded form that make use of the representation using dummy leaves. See Figure~\ref{fig:fcns-example} for an example. 

\begin{figure}[h]
\centerline{\scalebox{0.7}{
\begin{tikzpicture}[node distance = 0.5 cm]
\tikzstyle{VertexStyle} = [shape = rectangle, color = white, fill = white, text = black, draw]
\Vertex[x=4, y=4, L=$r$]{r} 
\Vertex[x=2.5, y=3.2,L=$b$]{b1}
\Vertex[x=3.5, y=3.2, L=$b$]{b2}
\Vertex[x=4.5, y=3.2, L=$b$]{b3}
\Vertex[x=5.5, y=3.2, L=$c$]{c1}
\Vertex[x=1.9, y=2.4, L=$a$]{a1}
\Vertex[x=2.5, y=2.4, L=$a$]{a2}
\Vertex[x=3.1, y=2.4, L=$c$]{c2}
\Vertex[x=4.2, y=2.4, L=$a$]{a3}
\Vertex[x=4.8, y=2.4, L=$a$]{a4}  
\Vertex[x=4.8, y=1.5, L=$ $]{s}  
\Edge(r)(b1)
\Edge(r)(b2)
\Edge(r)(b3)
\Edge(r)(c1)
\Edge(b1)(a1)
\Edge(b1)(a2)
\Edge(b1)(c2)
\Edge(b3)(a3)   
\Edge(b3)(a4)
\end{tikzpicture}
\hspace{0.01cm}
\begin{tikzpicture}[node distance = 0.5 cm]
\tikzstyle{VertexStyle} = [shape = rectangle, color = white, fill = white, text = black, draw]
	\Vertex[x=4, y=6, L=$r$]{r} 
	\Vertex[x=2.5, y=5.2, L=$b$]{b1}
	\Vertex[x=1.6, y=4.5, L=$a$]{a1}
	\Vertex[x=3.4, y=4.5, L=$b$]{b2}
	\Vertex[x=2.1, y=3.9, L=$a$]{a2}
	\Vertex[x=2.5, y=3.3, L=$c$]{c}
	\Vertex[x=3.9, y=3.9, L=$b$]{b3}
	\Vertex[x=3.4, y=3.3, L=$a$]{a3}
	\Vertex[x=4.3, y=3.3, L=$c$]{c2}
	\Vertex[x=3.8, y=2.8, L=$a$]{a4}
	\Edge(r)(b1)
	\Edge(b1)(a1)
	\Edge(b1)(b2)
	\Edge(a1)(a2)
	\Edge(a2)(c)
	\Edge(b2)(b3)
	\Edge(b3)(a3)
	\Edge(b3)(c2)
	\Edge(a3)(a4)
\end{tikzpicture} \begin{tikzpicture}[node distance = 0.5 cm]
\tikzstyle{VertexStyle} = [shape = rectangle, color = white, fill = white, text = black, draw]
	\Vertex[x=4, y=6, L=$r$]{r} 
	\Vertex[x=2.5, y=5.2, L=$b$]{b1}
	\Vertex[x=1.6, y=4.5, L=$a$]{a1}
	\Vertex[x=3.4, y=4.5, L=$b$]{b2}
	\Vertex[x=2.1, y=3.9, L=$a$]{a2}
	\Vertex[x=2.5, y=3.3, L=$c$]{c}
	\Vertex[x=3.9, y=3.9, L=$b$]{b3}
	\Vertex[x=3.4, y=3.3, L=$a$]{a3}
	\Vertex[x=4.3, y=3.3, L=$c$]{c2}
	\Vertex[x=3.8, y=2.8, L=$a$]{a4}
	\Vertex[x=4.5, y=5.5, L=$\bot$]{dummy1}
	\Vertex[x=1.1, y=4, L=$\bot$]{dummy2}
	\Vertex[x=1.6, y=3.4, L=$\bot$]{dummy3}
	\Vertex[x=2.9, y=4, L=$\bot$]{dummy4}
	\Vertex[x=2.9, y=2.8, L=$\bot$]{dummy5}
   	\Edge(r)(dummy1)
	\Edge(a1)(dummy2)
	\Edge(a2)(dummy3)
	\Edge(b2)(dummy4)
	\Edge(a3)(dummy5)            
	\Edge(r)(b1)
	\Edge(b1)(a1)
	\Edge(b1)(b2)
	\Edge(a1)(a2)
	\Edge(a2)(c)
	\Edge(b2)(b3)
	\Edge(b3)(a3)
	\Edge(b3)(c2)
	\Edge(a3)(a4)
\end{tikzpicture}
}}
\caption{Left: introductory example tree $t$ already shown in Figure~\ref{fig:xml-example}. Middle: \ac{FCNS} encoding of $t$: 
$\XML(\FCNS(t))=r_Lb_La_La_Rc_R\cl{c_R}\cl{a_R}\cl{a_L}b_Rb_R a_L a_R  \cl{a_R}\cl{a_L}c_R\cl{c_R}\cl{b_R}\cl{b_R}\cl{b_L}\cl{r_L}$. Right: $\FCNS^{\bot}$ encoding of $t$:
$\XML(\FCNS^{\bot}(t))=rba\bot\cl{\bot}a\bot\cl{\bot}c\cl{c}\cl{a}\cl{a}b\bot\cl{\bot}ba\bot\cl{\bot}a  \cl{a	}\cl{a}c\cl{c}\cl{b}\cl{b}\cl{b}\bot\cl{\bot}\cl{r}$.  \label{fig:fcns-example}}
\end{figure}
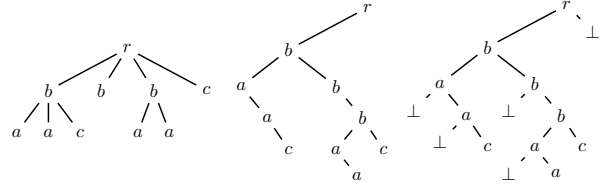

In the following subsections we provide streaming algorithms for the transformation of $\XML(t)$ to $\XML(\FCNS(t))$, that we call
the {\em \ac{FCNS} encoding}, and its inverse, the {\em \ac{FCNS} decoding}.

The \ac{FCNS} encoding can be seen as a reordering of the tags of $\XML(t)$ and an annotation of the tags with left/right. 
We state several properties about the relationship of the ordering of the tags in $\XML(t)$ and $\XML(\FCNS(t))$. 
Fact~\ref{fact_opening_tags_order} concerns the structure of the subsequence of opening tags in $\XML(\FCNS(t))$, Fact~\ref{fact_closing_tags_order} concerns the structure of the subsequence of closing tags in $\XML(\FCNS(t))$, and 
Fact~\ref{fact_opening_closing_tags} concerns the interplay of the subsequences of opening and closing tags in $\XML(\FCNS(t))$.
\begin{fact}
\label{fact_opening_tags_order}
The opening tags in $\XML(t)$ are in the same order as the opening tags in $\XML(\FCNS(t))$.
\end{fact}
For a node $v$ of some tree $t$, let $\timee'({v})$ and $\timee'(\bar{v})$ be the respective positions of the opening and closing tags 
of $v$ in $\XML(\FCNS(t))$.
\begin{fact}
\label{fact_closing_tags_order}
Nodes $v_1, v_2 $ of $t$ satisfy $\timee'( \cl{v_1} ) <  \timee'( \cl{v_2})$ iff one of the following conditions holds:
\begin{enumerate}
\item $v_1$ is in the subtree of $v_2$ in $t$; \label{fact_closing_tags_order_1}
\item or $v_1$ is a right sibling of $v_2$ in $t$; \label{fact_closing_tags_order_2}
\item or there is a node $u$ with $\depth(u) \le \depth(v_1) - 2$ such that $\timee(v_1) < \timee(u) \le \timee(v_2)$. \label{fact_closing_tags_order_3}
\end{enumerate}
\end{fact}

\begin{fact}
\label{fact_opening_closing_tags}
Nodes $v_1, v_2$ of $t$ satisfy $\timee'( \cl{v_1} ) <  \timee'({v_2})$ iff 
there is a node $u$ with $\depth(u) \le \depth(v_1) - 2$ such that $\timee(v_1) < \timee(u) \le \timee(v_2)$.
\end{fact}


\subsection{\ac{FCNS} encoding}
In this section, we are interested in computing the transformation $\XML(t) \rightarrow \XML(\FCNS(t))$. Our strategy is to compute the subsequence of opening tags of $\XML(\FCNS(t))$ (using Fact~\ref{fact_opening_tags_order} and discussed in subsection \ref{sec:seq_opening_tags}) and the subsequence of closing tags (using Fact~\ref{fact_closing_tags_order} and discussed in \ref{sec:seq_closing_tags}) of $\XML(\FCNS(t))$ independently, and merge them afterwards (using Fact~\ref{fact_opening_closing_tags} and discussed in subsection \ref{sec:seq_opening_closing_tags}).

\subsubsection{Computing the sequence of opening tags} \label{sec:seq_opening_tags}
Concerning the opening tags, since due to Fact~\ref{fact_opening_tags_order} the subsequences of opening tags in $\XML(t)$ and $\XML(\FCNS(t))$ coincide, we extract the subsequence of opening tag of $\XML(t)$, and we annotate them with left or right as they should be in $\XML(\FCNS(t))$. Remind that an opening tag is left if it is the opening tag of a first child, otherwise it is right. Furthermore, for later use we annotate each opening tag $c$ with $\depth(c)$ in $t$ and the position in the stream $\timee(c)$. We summarize this as a fact:

\begin{fact} \label{fact:seq_of_opening_tags}
There is a streaming algorithm with space $\Order(\log N)$ that, given $\XML(t)$ as input, outputs on an auxiliary stream the sequence of opening tags of $\XML(\FCNS(t))$ with left/right annotations, and furthermore, annotates each tag $c$ with $\depth(c)$ and $\pos(c)$, performing one pass on each stream.
\end{fact}

\subsubsection{Computing the sequence of closing tags}
\label{sec:seq_closing_tags}
For computing the sequence of closing tags, we start with the sequence of opening tags of $\XML(t)$ as produced by the output of the algorithm of Fact~\ref{fact:seq_of_opening_tags}, that is, correctly annotated with left/right and with depth and position annotations. To obtain the correct subsequence of closing tags as in $\XML(\FCNS(t))$, we interpret the opening tags as closing tags and we sort them with a merge sort algorithm. Merge sort can be implemented as a streaming algorithm with $\Order(\log(N))$ passes and 3 auxiliary streams~\cite{ghs09}. For the sake of simplicity, Algorithm~\ref{algo_merge_sort} assumes an input of length $2^l$ for some $l > 0$.

\begin{algorithm}[H]\small
\caption{Merge sort \label{algo_merge_sort}}
\begin{algorithmic}[1]
\REQUIRE unsorted data of length $2^l$ on stream 1
\FOR{$i=0\dots l-1$}
 \STATE copy data in blocks of length $2^i$ from stream 1 alternately onto stream 2 and stream 3
 \STATE \textbf{for} $j=1 \dots 2^{l-i-1}$ merge($2^i$) \textbf{end for}
\ENDFOR
\end{algorithmic}
\end{algorithm}

merge($b$) reads simultaneously the next $b$ values from stream 2 and stream 3, and merges them onto stream 1. The whole loop in Line~3 of Algorithm~\ref{algo_merge_sort} requires one read pass on stream 2, one read pass on stream 3, and one write pass on stream 1. See Figure~\ref{fig:merge_sort} for an illustration.

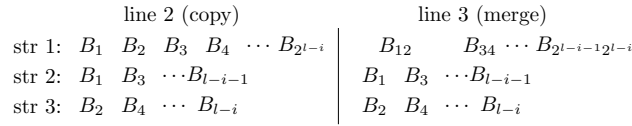
\begin{figure}[h]
\centerline{\scalebox{0.8}{
\begin{tikzpicture}[node distance = 0.5 cm]
\draw (4, 2.5) node {line $2$ (copy)};
\draw (6.6, 0.7) -- (6.6, 2.3);
\draw (9, 2.5) node {line $3$ (merge)};
\draw (1.6,2) node {str $1$:};
\draw (1.6,1.5) node {str $2$:};
\draw (1.6,1) node {str $3$:};
\draw (2.5, 2) node {$B_1$};\draw (3.2, 2) node {$B_2$};\draw (3.9, 2) node {$B_3$};\draw (4.6, 2) node {$B_4$};
\draw (5.3, 2) node {$\dots$};\draw (6, 2) node {$B_{2^{l-i}}$};
\draw (2.5, 1.5) node {$B_1$};
\draw (2.5, 1) node {$B_2$};
\draw (3.2, 1.5) node {$B_3$};
\draw (3.2, 1) node {$B_4$};
\draw (3.9, 1.5) node {$\dots$};
\draw (3.9, 1) node {$\dots$};
\draw (4.6, 1.5) node {$B_{l-i-1}$};
\draw (4.6, 1) node {$B_{l-i}$};
\draw (7.55, 2) node {$B_{12}$};\draw (8.95, 2) node {$B_{34}$};
\draw (9.6, 2) node {$\dots$};\draw (10.7, 2) node {$B_{2^{l-i-1} 2^{l-i}}$};
\draw (7.2, 1.5) node {$B_1$};
\draw (7.2, 1) node {$B_2$};
\draw (7.9, 1.5) node {$B_3$};
\draw (7.9, 1) node {$B_4$};
\draw (8.6, 1.5) node {$\dots$};
\draw (8.6, 1) node {$\dots$};
\draw (9.3, 1.5) node {$B_{l-i-1}$};
\draw (9.3, 1) node {$B_{l-i}$};
\end{tikzpicture} }}
\caption{In Line~2, blocks from stream 1 are copied onto stream 2 and stream 3. The $B_i$ are sorted blocks. In line $3$, all blocks $B_i$ and $B_{i+1}$ are merged into a sorted block $B_{i(i+1)}$.
\label{fig:merge_sort}}
\end{figure}

In order to use merge sort, we have to define a comparator function that, given two closing tags $\cl{c_1}, \cl{c_2}$, decides whether $\pos'(\cl{c_1}) < \pos'(\cl{c_2})$. Firstly, consider nodes $v_1, v_2$ with $\pos(v_1) < \pos(v_2)$ to be as in Point~1 or Point~2 of Fact~\ref{fact_closing_tags_order}, that is, either $v_1, v_2$ are siblings or one node is contained in the subtree of the other one. Evidently, their ordering with respect to $\pos'$ can easily be decided by their $\depth$: $\pos'(\cl{v_1}) < \pos'(\cl{v_2})$ iff $\depth(v_1) > \depth(v_2)$.

If neither $v_1, v_2$ are siblings, nor $v_2$ is in the subtree of $v_1$ (Point~3 of Fact~\ref{fact_closing_tags_order}), then $\pos'(\cl{v_1}) < \pos'(\cl{v_2})$, independently of their depths. A comparison function hence should be able to infer the relationship of the two nodes, however, this seems to be difficult in the streaming model.

To overcome this problem, instead of defining a comparison function, we design a complete merge function in Lemma~\ref{lemma:well-formed} that, by construction, only compares two nodes of the first kind. The key idea is to introduce \textit{separator} tags which we denote by a new tag outside $\Sigma$. They are initially inserted right after each closing tag of a last child $u$, that is exactly before the depth decreases. We denote by $\sep{u}$ the separator we introduce when seeing the last child $u$, and we define 
$\depth(\sep{u}) = \depth(u)$. 
\begin{fact}\label{fact:extract-opening}
There is a streaming algorithm with space $\Order(\log N)$ that, given a sequence $\XML(t)$ on a stream,
computes on another stream the sequence of opening tags $\XML(\FCNS(t))$ together with their separators, 
and annotated with $\depth$, $\timee$ and left/right, 
performing one pass on each stream.
\end{fact}
We have to define the way we integrate the separators into our sorting.
Let $v_1,v_2,\ldots,v_k$ be the ordered sequence of the children of some node. 
For the separator $\sep{v_k}$ we ask their position among the closing tags to satisfy for each node $v$:
\begin{equation}
\label{order_closing_tags}
\timee'(\cl{v}) < \timee'(\sep{v_k})  \quad \text{iff}\quad  \timee'(\cl{v}) \le \timee'(\cl{v_1});
\end{equation}
and for any other separator $\sep{w_k}$:
\begin{equation}
\label{order_closing_tags2}
\timee'(\sep{v_k}) < \timee'(\sep{w_k})  \quad \text{iff}\quad  \timee'(v_k) < \timee'(w_k) .
\end{equation}


Blocks appearing in merge sort fulfill a property that we call {\em well-sorted}. A block $B$ of closing tags is {\em well-sorted} if the corresponding tags in $\XML(\FCNS(t))$ appear in the same order, and for all $\cl{v_1}, \cl{v_2} \in B$ with $\pos(v_1) < \pos(v_2)$, all closing tags $\cl{v}$ of nodes $v$ with $\pos(v_1) < \pos(v) < \pos(v_2)$ are in $B$ as well.

In addition, for two blocks $B_1,B_2$ of closing tags, we say that $(B_1,B_2)$ is a {\em well-sorted adjacent pair},
if $B_1$ and $B_2$ are well-sorted, for each closing tag $\cl{v_1} \in B_1$ and each closing tag $\cl{v_2} \in B_2$ 
$\timee(v_1) < \timee(v_2)$ is satisfied, and furthermore, all closing tags $\cl{v}$ of nodes $v$ with $\pos(v_1) < \pos(v) < \pos(v_2)$ are either in $B_1$ or $B_2$.

The only function to design is a comparator deciding for two closing tags $\cl{v_1},\cl{v_2}$ from
a well-sorted adjacent pair $(B_1,B_2)$ whether $\pos'(\cl{v_1}) < \pos'(\cl{v_2})$.

The following lemma shows that we can merge a well-sorted adjacent pair correctly.

\begin{lemma}\label{lemma:well-formed}
Let $(B_1,B_2)$ be a well-sorted adjacent pair, and let $v_1 = B_1[p_1]$ and $v_2 = B_2[p_2]$ for some $p_1, p_2$. 
Assume that $\pos'(v) < \pos'(v_1)$ and $\pos'(v) < \pos'(v_2) $, for all $v \in B_1[1, p_1-1] \cup B_2[1, p_2-1]$. 
Then:
\begin{enumerate}
\item If $v_1$ is a separator, or there is a separator in $B_1$ after $v_1$, then $\timee'(v_1)<\timee'(v_2)$;
\item Else if $v_2$ is a separator then:
\begin{enumerate}
\item if $\depth(v_1) < \depth(v_2)$ then $\timee'(v_1)<\timee'(v_2)$,
\item else $\timee'(v_1) > \timee'(v_2)$;
\end{enumerate}

\item Else (neither $v_1$ nor $v_2$ is a separator):
\begin{enumerate}
\item if $\depth(v_1)\leq \depth(v_2)$ then $\timee'(v_1) < \timee'(v_2)$,
\item else $\timee'(v_1)>\timee'(v_2)$.
\end{enumerate}
\end{enumerate}
\end{lemma}
\begin{proof}
Let $(B_1,B_2)$ be a well-sorted adjacent pair. Let \mbox{$l = \max \{ i: B_1[i] \mbox{ is a separator} \}$}. If there are no separators in $B_1$, let $l = 0$.
First, we prove Point~1. Since $B_1$ is well-ordered, we only need to check that $\timee'(B_1\left[ l \right]) < \timee'(B_2\left[ 1 \right])$. Denote by $u$ the last child that was responsible for the insertion of the separator tag $B_1[l]$. Let $u'$ be the left-most sibling of $u$.
Due to Equation~\eqref{order_closing_tags} it suffices to show that $\timee'(\cl{u'}) < \timee'(B_2[1])$. 
Clearly, the shortest path from $u'$ to $B_2[1]$ passes by a common ancestor $p$ of $u'$ and $B_2[1]$ which is not the parent of $u'$ since the separator $B_1[l]$ indicates that the last child $u$ has been seen. Then, by the third condition of Fact~\ref{fact_closing_tags_order}, we get $\timee'(\cl{u'}) < \timee'(B_2[1])$.

For proving Points~2~and~3 we use the observation that if the premises to Point~1 are not fulfilled, $v_1, v_2$ do not have a common ancestor $p$ s.t. $\timee(v_1) < \timee(p) < \timee(v_2)$ and $p$ is not the parent node of $v_1$. Furthermore, this observation implies that $\depth(v_2) \ge \depth(v_1) - 1$ and hence, if $\depth(v_2) > \depth(v_1)$ then $v_2$ is in the subtree of $v_1$. This and Fact~\ref{fact_closing_tags_order} prove 
Points~2a, 2b, 3a and 3b.

We prove the observation by contradiction. Assume that there is such a node $p$. 
Since $(B_1, B_2)$ is a well-ordered adjacent pair and $\timee(v_1) < \timee(p) < \timee(v_2)$, node $p$ would be in $B_1 \cup B_2$.
Therefore, the separator $\sep{u}$ inserted after the rightmost sibling of $v_1$ would be also in $B_1 \cup B_2$ as well. 
More precisely, this separator would be in $B_2[1 \dots p_2-1]$ since otherwise Point~1 would have been applied. This, however, is a contradiction to the assumption that $\pos'(v) < \pos'(v_1) \, \forall v \in B_1[1 \dots p_1 - 1] \cup B_2[1 \dots p_2 - 1]$ since it holds that $\pos'(v_1) < \pos'(\sep{u})$. Hence such a node does not exist.
\end{proof}

\begin{lemma}
\label{lem:closing_tags}
There is a $\Order(\log N)$-pass streaming algorithm with space $\Order(\log N)$ and $3$ auxiliary streams that computes the subsequence of closing tags of the \ac{FCNS} encoding of any XML document given in the input stream.
\end{lemma}

\begin{proof}
Using Fact~\ref{fact:extract-opening}, we compute on the first auxiliary stream the sequence of opening tags with corresponding annotations, together with separators, and interpret opening tags as closing tags. 

We show that we can do a merge sort algorithm with a merge function inspired by Lemma~\ref{lemma:well-formed} on the first three auxiliary streams with $\Order(\log N)$ space and passes.
For that assume that the first stream contains a sequence $(B_1,B_2,\ldots,B_M)$ of blocks of size $2^i$. For simplicity we assume that $M$ is even, otherwise we add an empty block.
We alternately copy odd blocks on the second stream, and even blocks on the third stream.
For a block $B_{2i}$ that we write on the third stream, we write before each of them, the number of separators that occur in the block $B_{2i-1}$ that was copied on the second stream.

Then we merge sequentially all pairs of blocks $(B_{2k-1},B_{2k})$ for $1\leq k\leq M/2$ using Lemma~\ref{lemma:well-formed}. Note that
$(B_{2k-1},B_{2k})_i$ are all well-formed pairs. Let \mbox{$l = \max \{ i: B_{2k-1}[i] \mbox{ is a separator} \}$}. Firstly, we copy elements $B_{2k-1}[1, l]$ onto auxiliary stream 1. Knowing the number of separators in $B_{2k-1}$ allows us to perform this operation. The correctness of this step follows from Point~1 of Lemma~\ref{lemma:well-formed}. Then, we merge blocks $B_{2k-1}[l+1, 2^i]$ and $B_{2k}$ by using the comparison function defined in Points~2 and 3 of Lemma~\ref{lemma:well-formed}.
\end{proof}

\subsubsection{Merging opening and closing tags}
\label{sec:seq_opening_closing_tags}
Merging the subsequence of opening tags of $\XML(\FCNS(t))$ and the subsequence of closing tags of $\XML(\FCNS(t))$ can be done using one additional pass.
\begin{lemma}\label{lemma:fcns-unify}
There is a streaming algorithm with space $\Order(\log n)$ that, given the sequence of opening tags of $\XML(\FCNS(t))$ on a stream, 
and the sequence of closing tags of $\XML(\FCNS(t))$ on another stream, computes $\XML(\FCNS(t))$ on a third stream
using one pass on each stream.
\end{lemma}
\begin{proof}
We directly apply Fact~\ref{fact_opening_closing_tags}, so that we know when to alternate from the sequence of opening tags to the one of closing tags, and conversely.
\end{proof}

From Fact~\ref{fact:seq_of_opening_tags}, Lemma~\ref{lem:closing_tags} and Lemma~\ref{lemma:fcns-unify} we obtain Theorem~\ref{theorem:encoding}.

\begin{theorem}
\label{theorem:encoding}
There is a $\Order(\log N)$-pass streaming algorithm with space $\Order(\log N)$ and $3$ auxiliary streams and $\Order(1)$ processing time per letter that computes
on the third auxiliary stream the \ac{FCNS} transformation of any XML document given in the input stream.
\end{theorem}

\begin{proof}
Firstly, we compute according to Lemma~\ref{lem:closing_tags} the sequence of closing tags and we store them on auxiliary stream $1$. Then, by Fact~\ref{fact:seq_of_opening_tags} we extract the sequence of opening tags, and we store them on auxiliary stream $2$. By Lemma~\ref{lemma:fcns-unify} we can merge the tags of auxiliary stream 1 and auxiliary stream 2 correctly onto stream $3$.

The space requirements of these operations do not exceed $\Order(\log N)$. The processing time per letter of these operations is constant. 
\end{proof}

\subsection{Checking Validity on the encoding form}
\label{sec:validity_of_encoded_form}
In this section, we reuse the algorithms for validating binary trees for the validation of the encoded form. We discuss one-pass read/write streaming algorithms (one for left-to-right passes, and one for right-to-left passes) that read $\XML(\FCNS^{\bot}(t))$ and output an XML document with annotations on closing tags that can be fed into Algorithm~\ref{algo_bin_tree_one_pass} or Algorithm~\ref{algo_bin_tree_two_passes}. This requires little modifications in Algorithm~\ref{algo_bin_tree_one_pass} and Algorithm~\ref{algo_bin_tree_two_passes} since validity then depends on the annotations. Since we want to reuse the algorithms for the validation of binary trees, we suppose that $\XML(\FCNS^{\bot}(t))$ is available as input. The algorithm stated in Theorem~\ref{theorem:encoding} can be easily adapted such that it outputs $\XML(\FCNS^{\bot}(t))$ instead of $\XML(\FCNS(t))$.

The problem of validating the encoded form and the problem of validating binary trees are similar. Note that the children $v_1, \dots v_k$ of a node $v$ form a substring in $\XML(\FCNS^{\bot}(t))$, see Figure~\ref{figure:evaluating-fcns-form}. Hence, for validating a node $v$, the label $v$ has to be related to the \textit{block} of children nodes $\cl{v_k}\dots \cl{v_1}$. This is similar to the task of the validation of binary trees where the parent label $v$ has to be related to the block of children nodes $\cl{v_1}v_2$.

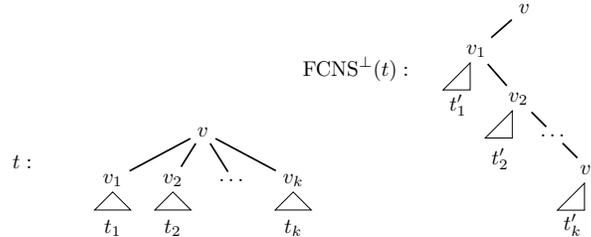
\begin{figure}[h]
\centerline{\scalebox{0.8}{
\begin{tikzpicture}[node distance = 0.5 cm]
\tikzstyle{VertexStyle} = [shape = rectangle, color = white, fill = white, text = black, draw]
\draw (1, 3.5) node {$t:$};
\Vertex[x=4, y=4, L=$v$]{v} 
\Vertex[x=2.5, y=3.2,L=$v_1$]{v1}
\Vertex[x=3.5, y=3.2, L=$v_2$]{v2}
\Vertex[x=4.5, y=3.2, L=$\dots$]{n}
\Vertex[x=5.5, y=3.2, L=$v_k$]{vk}
\Edge(v)(v1)
\Edge(v)(v2)
\Edge(v)(n)
\Edge(v)(vk)
\draw (2.5, 3) -- (2.2, 2.7) -- (2.8, 2.7) -- cycle;
\draw (3.5, 3) -- (3.2, 2.7) -- (3.8, 2.7) -- cycle;
\draw (5.5, 3) -- (5.2, 2.7) -- (5.8, 2.7) -- cycle;
\draw (2.5,2.4) node {$t_1$};
\draw (3.5,2.4) node {$t_2$};
\draw (5.5,2.4) node {$t_k$};
\end{tikzpicture}
\hspace{-0.5cm}
\begin{tikzpicture}[node distance = 0.5 cm]
\tikzstyle{VertexStyle} = [shape = rectangle, color = white, fill = white, text = black, draw]
\draw (1.2, 3) node {$\FCNS^{\bot}(t):$};
\Vertex[x=4, y=4, L=$v$]{v} 
\Vertex[x=3.2, y=3.3,L=$v_1$]{v1}
\Vertex[x=3.9, y=2.5, L=$v_2$]{v2}
\Vertex[x=4.5, y=1.9, L=$\dots$]{n}
\Vertex[x=5.1, y=1.3, L=$v_k$]{vk}
\Edge(v)(v1)
\Edge(v1)(v2)
\Edge(v2)(n)
\Edge(n)(vk)
\draw (3.1, 3.1) -- (2.65, 2.65) -- (3.1, 2.65) -- cycle;
\draw (3.8, 2.3) -- (3.35, 1.85) -- (3.8, 1.85) -- cycle;
\draw (5.0, 1.1) -- (4.55, 0.65) -- (5.0, 0.65) -- cycle;
\draw (2.9,2.4) node {$t_1'$};
\draw (3.6,1.5) node {$t_2'$};
\draw (4.8,0.4) node {$t_k'$};
\end{tikzpicture}
}}
\caption{A tree $t$ and its $\FCNS^{\bot}$ encoding. While the opening and closing tags of the children of a node $v$ are separated by the subtrees $t_1, \dots t_k$ in $\XML(t)$, the closing tags of the children of $v$ are consecutive in $\XML(\FCNS^{\bot}(t))$ in reverse order, that is $\cl{v_{k}}\cl{v_{k-1}}\dots\cl{v_{2}}\cl{v_{1}}$ is a substring of $\XML(\FCNS^{\bot}(t))$. \label{figure:evaluating-fcns-form}}
\end{figure}

For a node $v$, we gather the information of the children nodes $v_1, \dots, v_k$ by the help of finite automata $A_1$ (for left-to-right passes) and $A_2$ (for right-to-left passes) that we define later, and the information - a state of the automata - is annotated at the closing tag of leaf $v_1$ (left-to-right) or $v_k$ (right-to-left). Then, by the help of Algorithm~\ref{algo_bin_tree_one_pass} or Algorithm~\ref{algo_bin_tree_two_passes}, this information is related to the parent label $v$. 

We define automata $A_1$ and $A_2$ now. $A_1$ is constructed from automaton $A$. Let $\mbox{$A = (\Sigma, Q, q_0, \delta, F)$}$ be a deterministic finite automaton where $\Sigma$ is its input alphabet, $Q$ is the state set, $q_0$ is its initial state, $\delta: Q \times \Sigma \rightarrow Q$ is the transition function, and $F$ is a set of final states. For $a \in \Sigma$ and the input DTD $D$, denote by $A_a$ a deterministic finite automaton that accepts the regular expression $D(a)$. We compose the $A_a$ as in the left illustration of Figure~\ref{figure:automaton} to an automaton $A$ that accepts words $\omega'$ such that $\omega' = a \omega$, $a \in \Sigma, \omega \in \Sigma^*$ if $\omega \in D(a)$. 

\begin{figure}[h!]
\includegraphics[scale=0.4, bb=90 500 330 680]{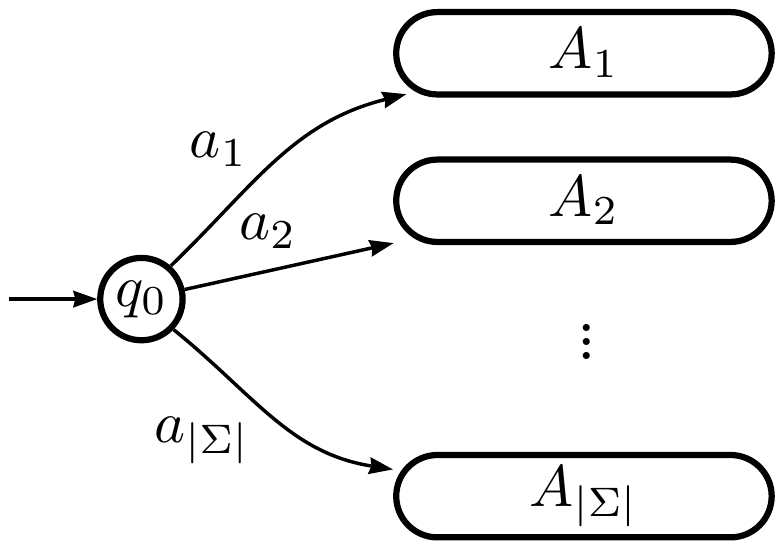} \includegraphics[scale=0.4, bb=100 500 400 680]{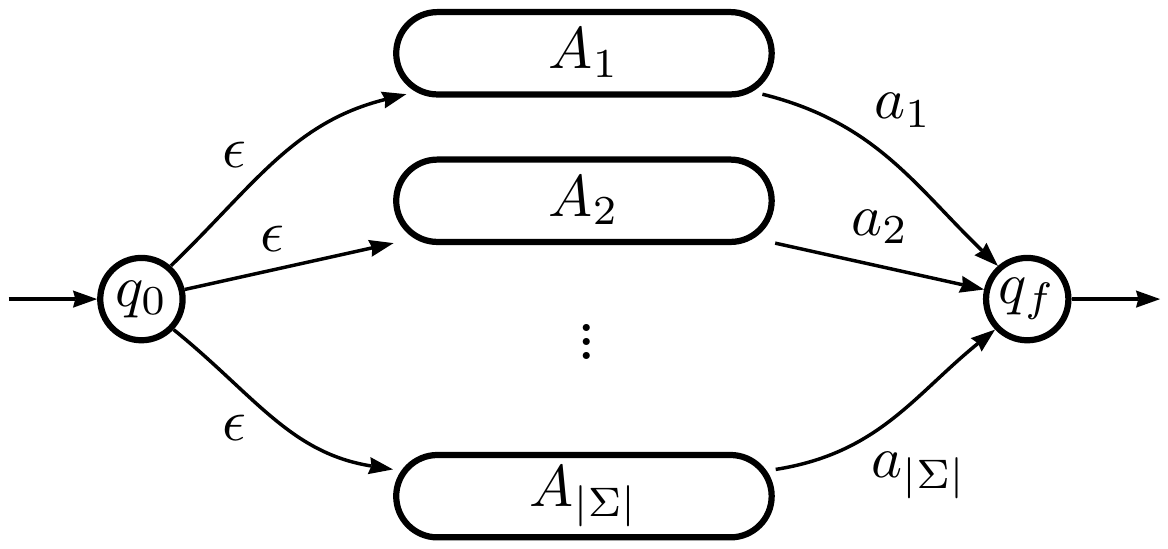}
\caption{Left: Automaton $A$. $A_1$ accepts words $\omega$ if $A$ accepts $\omega^{\rev}$. Right: Automaton $A_2$ is a version of the illustrated automaton without $\epsilon$ transitions. \label{figure:automaton}}
\end{figure}

Let $\mbox{$A_1 = (\Sigma, Q_1, (q_0)_1, \delta_1, F_1)$}$ be a deterministic finite automaton that accepts a word $\omega$, iff $\omega^{\rev}$ is accepted by $A$, where $\omega^{\rev}$ denotes $\omega$ read from right to left. 

Let $\mbox{$A_2 = (\Sigma, Q_2, (q_0)_2, \delta_2, F_2)$}$ be a deterministic finite automaton that accepts a word $\omega'$ such that $\omega' = \omega a$, $a \in \Sigma, \omega \in \Sigma^*$ if $\omega \in D(a)$. $A_2$ is a version of the automaton in the right illustration of Figure~\ref{figure:automaton} without $\epsilon$ transitions.

In the following, we assume for some state $q_1 \in Q_1$ and $q_2 \in Q_2$ that $\delta_1(q_1, \bot) = q_1$ and $\delta_2(q_2, \bot) = q_2$.

Given $\XML(\FCNS^{\bot}(t))$, for a left-to-right pass, we annotate closing tags $\cl{v}$ by a state from the state set $Q_1$ of automaton $A_1$. We denote the annotation for left-to-right passes of $\cl{v}$ by $\ann_1(\cl{v})$. 
\vspace{0.2cm} \\
\noindent \textbf{if} $v$ is a leaf then $\ann_1(\cl{v}) = \delta_1((q_0)_1, v)$, \\
\noindent \textbf{otherwise} let $u$ be the right child of $v$, then $\ann_1(\cl{v}) = \delta_1(\ann_1(\cl{u}), v)$. 
\vspace{0.2cm}

For a right-to-left pass, we annotate closing tags $\cl{v}$ by a state form the state set $Q_2$ of automaton $A_2$. We denote the annotation of right-to-left passes of $\cl{v}$ by $\ann_2(\cl{v})$. \vspace{0.2cm} \\
\noindent \textbf{if} $v$ is a left child then $\ann_2(\cl{v}) = \delta_2((q_0)_2, v)$, \\
\noindent \textbf{if} $v$ is a right child of $u$ then $\ann_2(\cl{v}) = \delta_2(\ann_2(\cl{u}), v)$ .
\vspace{0.2cm}

For the sake of completeness, the root can be annotated by $\ann_2(\cl{r}) = (q_0)_2$, though this annotation will not be used for checking validity. Figure~\ref{fig:r1} shows the annotations of the children nodes $v_1, \dots, v_k$ of node $v$.

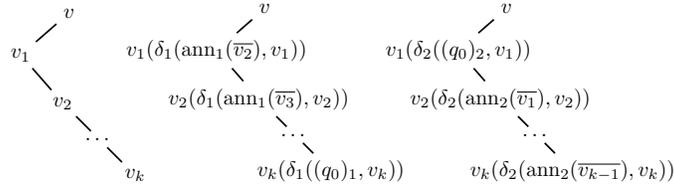
\begin{figure}[h]
\centerline{\scalebox{0.8}{
\begin{tikzpicture}[node distance = 0.5 cm]
\tikzstyle{VertexStyle} = [shape = rectangle, color = white, fill = white, text = black, draw]
\Vertex[x=4, y=4, L=$v$]{v} 
\Vertex[x=3.2, y=3.3,L=$v_1$]{v1}
\Vertex[x=3.9, y=2.5, L=$v_2$]{v2}
\Vertex[x=4.5, y=1.9, L=$\dots$]{n}
\Vertex[x=5.1, y=1.3, L=$v_k$]{vk}
\Edge(v)(v1)
\Edge(v1)(v2)
\Edge(v2)(n)
\Edge(n)(vk)
\end{tikzpicture} \hspace{-0.8cm}
\begin{tikzpicture}[node distance = 0.5 cm]
\tikzstyle{VertexStyle} = [shape = rectangle, color = white, fill = white, text = black, draw]
\Vertex[x=4, y=4, L=$v$]{v} 
\Vertex[x=3.2, y=3.3,L=$\mbox{$v_1(\delta_1(\ann_1(\cl{v_2}), v_1))$}$]{v1}
\Vertex[x=3.9, y=2.5, L=$\mbox{$v_2(\delta_1(\ann_1(\cl{v_3}), v_2))$}$]{v2}
\Vertex[x=4.5, y=1.9, L=$\dots$]{n}
\Vertex[x=5.1, y=1.3, L=$\mbox{$v_k(\delta_1((q_0)_1, v_k))$}$]{vk}
\Edge(v)(v1)
\Edge(v1)(v2)
\Edge(v2)(n)
\Edge(n)(vk)
\end{tikzpicture} \hspace{-0.8cm}
\begin{tikzpicture}[node distance = 0.5 cm]
\tikzstyle{VertexStyle} = [shape = rectangle, color = white, fill = white, text = black, draw]
\Vertex[x=4, y=4, L=$v$]{v} 
\Vertex[x=3.2, y=3.3,L=$\mbox{$v_1(\delta_2((q_0)_2, v_1))$}$]{v1}
\Vertex[x=3.9, y=2.5, L=$\mbox{$v_2(\delta_2(\ann_2(\cl{v_1}), v_2))$}$]{v2}
\Vertex[x=4.5, y=1.9, L=$\dots$]{n}
\Vertex[x=5.1, y=1.3, L=$\mbox{$v_k(\delta_2(\ann_2(\cl{v_{k-1}}), v_k))$}$]{vk}
\Edge(v)(v1)
\Edge(v1)(v2)
\Edge(v2)(n)
\Edge(n)(vk)
\end{tikzpicture}
}}
\caption{Left: a node $v$ with its children $v_1, \dots v_k$ in the FCNS tree. Middle: annotations for left-to-right passes. $v$ is valid if $\delta_1(\ann_1(\cl{v_1}), v)$ results in an accepting state of $A_1$. Right: annotations for right-to-left passes. $v$ is valid if $\delta_2(\ann_2(\cl{v_k}), v)$ is an accepting state of $A_2$. \label{fig:r1}}
\end{figure}

The annotation operations can be seen as streaming algorithms performing one read pass over the input and one write pass over another stream using constant space, since the annotation of a closing tag $\cl{v}$ only depends on the annotation of its right child (for left-to-right passes) or its parent (for right-to-left passes). The respective closing tag is in in both cases the tag prior to $v$ in the input stream $\XML(\FCNS^{\bot}(t))$. Remember that we consider a right-to-left pass for the annotation with $\ann_2$.




In the following, we prove that given the annotations, Algorithm~\ref{algo_bin_tree_one_pass} and Algorithm~\ref{algo_bin_tree_two_passes} can be adapted to decide validity of the encoded form.

\begin{theorem}
\label{thm:rabin-one-pass}
There is a one-pass deterministic algorithm for $\pb$ with space $\Order(\sqrt{N \log N})$ and $\Order(1)$ processing time per letter when the input is given in its $\FCNS^{\bot}$ encoding.
\end{theorem}

\begin{proof}
Append the rule $D(\bot) = \epsilon$ to the input DTD $D$. Compute automaton $A_1$. Compute a new XML stream with annotations $\ann_1$ and feed this stream directly into Algorithm~\ref{algo_bin_tree_one_pass}. In order to verify a node $v$, Algorithm~\ref{algo_bin_tree_one_pass} uses the closing tag $\cl{v_1}$ of the left child $v_1$ of $v$. Since $v$ is valid if $\delta_1(\ann_1(\cl{v_1}), v)$ is an accepting state, we only have to replace the check function used in Algorithm~\ref{algo_bin_tree_one_pass}. The new check function computes $\delta_1(\ann_1(\cl{v_1}), v)$ and aborts if the resulting state is not accepting.

The space requirements and the processing time per letter inherit from Algorithm~\ref{algo_bin_tree_one_pass}.
\end{proof}

\begin{theorem}
\label{thm:rabin-two-pass}
There is a bidirectional two-pass deterministic algorithm for $\pb$ with space $\Order(\log^2 N)$ and $\Order(\log N)$ processing time per letter when the input is given in its FCNS encoding.
\end{theorem}

\begin{proof}
Append the rule $D(\bot) = \epsilon$ to the input DTD $D$. Compute automata $A_1$ and $A_2$. Firstly, in a left-to-right pass, as in the proof of Theorem~\ref{thm:rabin-one-pass} we compute a new XML stream with annotations $\ann_1$ and feed this stream directly into Algorithm~\ref{algo_bin_tree_subroutine}. We adapt the check function as above.

Concerning the right-to-left pass, we compute the annotations $\ann_2$, and feed this stream directly into Algorithm~\ref{algo_bin_tree_subroutine} interpreting opening tags as closing tags, and vice versa. Note that the annotations $\ann_2$ are hence annotated on opening tags. Let $v$ be a node with children $v_1, \dots, v_k$. We have to show how Algorithm~\ref{algo_bin_tree_subroutine} can be adapted to relate the annotation of the opening tag $v_k$ to $v$. When Algorithm~\ref{algo_bin_tree_subroutine} reads the closing tag $\cl{w}$ and the opening tag $v_1$, it pushes $(\cl{w}, v_1, \cdot, \cdot)$ on the stack, where $w$ is the sibling of $v_1$ in the $\FCNS^{\bot}$ encoding. Note that the subsequent tags on the stream are $v_2, v_3, \dots v_k$. Node $v_k$ can be identified since $v_k$ is either a leaf or followed by a tag with label $\bot$. Hence, the stack item $(\cl{w}, v_1, \cdot, \cdot)$ can be annotated with $\ann_2(v_k)$ when $v_k$ is seen. Again by adapting the check routine, we can compute $\delta_2(\ann_2(v_k), v_1)$ and abort if the result is not an accepting state.

The space requirements and the processing time per letter inherit from Algorithm~\ref{algo_bin_tree_two_passes}.
\end{proof}

Applying the bidirectional algorithm of Theorem~\ref{thm:rabin-two-pass} on the encoded form $\XML(\FCNS^{\bot}(t))$, we obtain that validity of general trees can be decided memory efficiently in the streaming model with auxiliary streams.

\begin{corollary}
\label{corollary:main-result}
There is a bidirectional $\Order( \log N)$-pass deterministic streaming algorithm for $\pb$ with space $\Order(\log^2 N)$, $\Order(\log N)$ processing time per letter, and $3$ auxiliary streams.
\end{corollary}

\subsection{Decoding}
In the following, we present a streaming algorithm for \ac{FCNS} decoding, that is, given $\XML(\FCNS(t))$ of some tree $t$, output $\XML(t)$. We start with a non-streaming algorithm, Algorithm~\ref{algo:non-straming-FCNS-to-unranked} performing this task. 
\begin{algorithm}[h!]\small
\caption{offline algorithm for \ac{FCNS} decoding}
\label{algo:non-straming-FCNS-to-unranked}
\begin{algorithmic}[1]
\FOR{$i = 1 \rightarrow 2N$}
\IF{$X[i]$ is an opening tag}
\STATE write $X[i]$
\IF{$X[i]$ does not have a left subtree} \label{line:easy_condition}
\STATE write $\cl{X[i]}$
\ENDIF
\ELSIF[See Figure~\ref{fig:decoding}]{$X[i]$ is a left closing tag}
\STATE let $p$ be the parent node of $X[i]$ \label{line:difficulty_conversion}
\STATE write $\cl{p}$
\ENDIF
\ENDFOR
\end{algorithmic}
\end{algorithm}
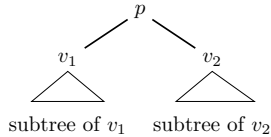
\begin{figure}[h]
\centerline{\scalebox{0.8}{
\begin{tikzpicture}[node distance = 0.5 cm]
\tikzstyle{VertexStyle} = [shape = rectangle, color = white, fill = white, text = black, draw]
\Vertex[x=4, y=4, L=$p$]{v} 
\Vertex[x=2.8, y=3.2,L=$v_1$]{v1}
\Vertex[x=5.2, y=3.2, L=$v_2$]{vk}
\Edge(v)(v1)
\Edge(v)(vk)
\draw (2.8, 3) -- (2.2, 2.5) -- (3.4, 2.5) -- cycle;
\draw (5.2, 3) -- (4.6, 2.5) -- (5.9, 2.5) -- cycle;
\draw (2.8, 2.1) node {subtree of $v_1$};
\draw (5.2, 2.1) node {subtree of $v_2$};
\end{tikzpicture}
}}
\caption{The main difficulty of the \ac{FCNS} decoding is to write the closing tag of a node $p$ when the closing tag of its left child is seen. This is difficult when the subtrees of $v_1$ and $v_2$ are large.\label{fig:decoding}}
\end{figure}


We describe how this algorithm can be converted into a streaming algorithm. For some opening tag $X[i]$, checking the condition in Line~\ref{line:easy_condition} can easily be done by investigating $X[i+1]$. If $X[i+1]$ is a right opening tag or equals $\cl{X[i]}$, $X[i]$ does not have a left subtree. The difficulty in converting this algorithm into a streaming algorithm is in Line~\ref{line:difficulty_conversion}, it is difficult to keep track of opening tags until the respective closing tags of their left children are seen, and indeed, this can not be done with sublinear space in one pass
(Fact~\ref{fact:lb-decoding}).

In the following, we present a streaming algorithm that performs one pass over the input, but two passes over the output, and uses $O(\sqrt{N\log{N}})$ space, and a streaming algorithm that performs $\Order( \log N)$ passes over the input and 3 auxiliary streams using $\Order(\log^2 (N))$ space.

\subsubsection{One read-pass and two write-passes}
We read blocks of size $\sqrt{N\log N}$ and execute Algorithm~\ref{algo:non-straming-FCNS-to-unranked} on that block. In Lemma~\ref{lemma:blocklemma} we shows that in any block there is at most one left closing tag for which the parent's opening and closing tag are not in that block. Hence per block there is at most one left closing tag for which we can not obtain the label of the parent node. We call this closing tag \textit{critical}. In this case we write a \textit{dummy symbol} on the output stream that will be overwritten by the parent closing tag in the second pass. The closing tag of the parent node will arrive in a subsequent block, and it can easily be identified as this since it is the next closing tag arriving at a depth $ - 1$ of the critical closing tag. We store it upon its arrival in our random access memory. Since there is at most one critical closing tag per block and we have a block size of $\sqrt{N\log N}$, we have to recover at most $\Order(\sqrt{N/\log N})$ parent nodes. At the end of the pass over the input stream we have recovered all closing tags of parent nodes for which we wrote dummy symbols on the output stream. In a second pass over the output stream we overwrite the dummy symbols by the correct closing tags. 

The complexity derives from the following lemma demonstrating that in a block there is at most one critical left closing tag. 
\begin{lemma}\label{lemma:blocklemma}
Let $X[i,j]$ be a block. Then there is at most one left closing tag $\cl{a}$ with parent node $p$ such that:
\begin{equation}
\timee(p)  <  i \le  \timee(\cl{a}) \le j  <  \timee(\cl{p}).  \label{ineq:blocklemma}
\end{equation}
\end{lemma}

\begin{proof}
By contradiction, assume that there are $2$ left closing tags $\cl{a}, \cl{b}$ with $p$ being the parent node of $a$, and $q$ being the parent node of $b$, for which Inequality~\ref{ineq:blocklemma} holds. Wlog we assume that $\timee(p) < \timee(q)$. Since $\timee(p) < \timee(q) < \timee(\cl{a})$, $q$ is contained in the subtree of $a$ or $q = a$. This, however, implies that $\timee(\cl{q}) \le \timee(\cl{a}) < j$ contradicting $\pos(\cl{q}) > j$.
\end{proof}

\begin{theorem}
\label{thm:decoding_two_pass}
There is a streaming algorithm using $\Order(\sqrt{N \log N})$ space and $\Order(1)$ processing time per letter which performs one pass over the input stream containing $\XML(t)$ and two passes over the output stream onto which it outputs $\XML(\FCNS(t))$.
\end{theorem}

\subsubsection{Logarithmic number of passes}


Again, we use the offline Algorithm~\ref{algo:non-straming-FCNS-to-unranked} as a starting point for the algorithm we design now. For coping with the problem that it is hard to remember all opening parent tags when their corresponding closing tag ought to be written on the output, we write categorically \textit{dummy symbols} on the output stream for all parent closing tags. The crux then is the following observation:

\begin{fact}
\label{fact:decoding_order}
Let $\cl{c_1}_L \dots \cl{c_N}_L$ be the subsequence of closing tags of left children of $\XML(\FCNS(t))$. Then the sequence $\cl{p_1} \dots \cl{p_N}$ is a subsequence of $\XML(t)$ where $p_i$ is the parent node of $c_i$ in $\FCNS(t)$.
\end{fact}

We apply a modified version of our bidirectional two-pass Algorithm~\ref{algo_bin_tree_two_passes} to recover the missing tags. Instead of checking validity, once the check function is called in Algorithm~\ref{algo_bin_tree_subroutine} with variables $(a,b,c)$, we output the parent label $a$ onto an auxiliary stream, annotated with $\pos(b)$. We do the same in a reverse pass over the input stream counting positions from $2N$ downwards to $1$. In so doing, the auxiliary stream contains all parent labels for which dummy symbols are written on the output stream. 

Fact~\ref{fact:decoding_order} shows that it is enough to sort by means of two further auxiliary streams the auxiliary stream with respect to the annotated position of the closing tags of the left children of these nodes. In a last pass we insert the parent closing tags into the output stream.





\begin{theorem}\label{thm:decoding_log_pass}
There is a $\Order(\log N)$-pass streaming algorithm with space $\Order(\log^2 N)$ and $\Order(\log N)$ processing time per letter and $3$ auxiliary streams that computes
on the third auxiliary stream the \ac{FCNS} decoding of any FCNS encoded document given in the input stream.
\end{theorem}

\section{Lower bounds for FCNS encoding and decoding}\label{sec:lowerbounds}

We define a family of hard instances of length $N=\Theta(n)$ for the computation of the \ac{FCNS} encoding of a tree as in Figure~\ref{fig:lb-encoding}.

\begin{figure}[h] 
\centerline{\scalebox{0.8}{ 
\begin{tikzpicture}[node distance = 0.5 cm]
\tikzstyle{VertexStyle} = [shape = rectangle, color = white, fill = white, text = black, draw]
\Vertex[x=4, y=4, L=$r$]{v} 
\Vertex[x=2.5, y=3.2,L=$\mbox{$x_1$}$]{v1}
\Vertex[x=3.5, y=3.2, L=$\mbox{$x_2$}$]{v2}
\Vertex[x=4.5, y=3.2, L=$\dots$]{n}
\Vertex[x=5.5, y=3.2, L=$\mbox{$x_n$}$]{vk}
\Edge(v)(v1)
\Edge(v)(v2)
\Edge(v)(n)
\Edge(v)(vk)
\end{tikzpicture}
\hspace{1cm}
\begin{tikzpicture}[node distance = 0.5 cm]
\tikzstyle{VertexStyle} = [shape = rectangle, color = white, fill = white, text = black, draw]
\Vertex[x=4, y=4, L=$r$]{v} 
\Vertex[x=3.2, y=3.3,L=$\mbox{$x_1$}$]{v1}
\Vertex[x=3.9, y=2.5, L=$\mbox{$x_2$}$]{v2}
\Vertex[x=4.5, y=1.9, L=$\dots$]{n}
\Vertex[x=5.1, y=1.3, L=$\mbox{$x_n$}$]{vk}
\Edge(v)(v1)
\Edge(v1)(v2)
\Edge(v2)(n)
\Edge(n)(vk)
\end{tikzpicture}
}}
\caption{\label{fig:lb-encoding}Left: hard instance. Right: its FCNS encoded form.}
\end{figure}
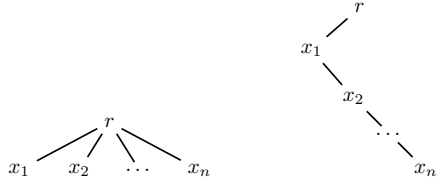

It is easy to see that computing the sequence of closing tags in the FCNS encoding requires to invert a stream. Let $t$ be a hard instance. Then $\mbox{$\XML(t) = rx_1\cl{x_1}x_2\cl{x_2} \dots x_n\cl{x_n}\cl{r}$}$, and $\XML(\FCNS(t)) = r_Lx_{1L}x_{2R} \dots x_{nR} \cl{x_{nR}}\cl{x_{n-1R}} \dots \cl{x_{2R}} \cl{x_{1L}} \cl{r_L}$.
Since the writing on the stream can only start after reading $x_n$, we deduce that memory space $\Omega(n)$ is required, in order to store all previous tags.
\begin{fact}\label{fact:lb-encoding}
Every one-pass randomized streaming algorithm for FCNS encoding
with bounded error requires $\Omega(N)$ space.
\end{fact}

We conjecture that this argument can be extended as follows:

\begin{conjecture} \label{conjecture:space-passes}
Any $p$-passes randomized streaming algorithm for FCNS encoding with bounded error requires space $\Omega(N/p)$.
\end{conjecture}

We now define another family of hard instances of length $N=\Theta(n)$ for decoding a \ac{FCNS} encoded tree as in Figure~\ref{fig:lb-decoding}.

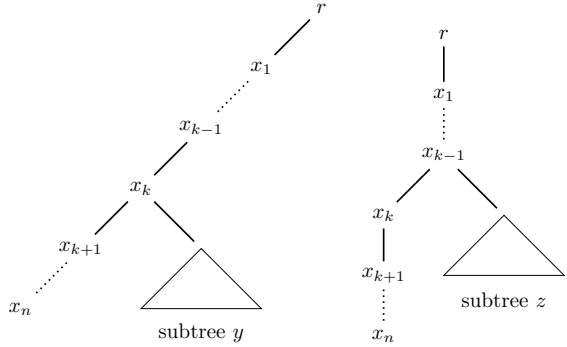
\begin{figure}[h] 
\centerline{\scalebox{0.8}{
\begin{tikzpicture}[node distance = 0.5 cm]
\tikzstyle{VertexStyle} = [shape = rectangle, color = white, fill = white, text = black, draw]
\Vertex[x=10, y=10, L=$r$]{r} 
\Vertex[x=9, y=9,L=$\mbox{$x_1$}$]{x1}
\Vertex[x=8, y=8,L=$\mbox{$x_{k-1}$}$]{x2}
\Vertex[x=7, y=7,L=$\mbox{$x_k$}$]{xk}
\Vertex[x=6, y=6,L=$\mbox{$x_{k+1}$}$]{xk1}
\Vertex[x=5, y=5,L=$\mbox{$x_{n}$}$]{xn}
\Vertex[x=8, y=6,L=$\mbox{$\,$}$]{dummy}
\Edge(r)(x1)
\Edge(x2)(xk)
\Edge(xk)(xk1)
\Edge(xk)(dummy)
\tikzset{EdgeStyle/.style={dotted}}
\Edge(x1)(x2)
\Edge(xk1)(xn)
\draw (8, 6) -- (7, 5) -- (9, 5) -- cycle;
\draw (8,4.6) node {subtree $y$};
\end{tikzpicture}
\hspace{0.1cm}
\begin{tikzpicture}[node distance = 0.5 cm]
\tikzstyle{VertexStyle} = [shape = rectangle, color = white, fill = white, text = black, draw]
\Vertex[x=10, y=10, L=$r$]{r} 
\Vertex[x=10, y=9, L=$\mbox{$x_1$}$]{x1} 
\Vertex[x=10, y=8, L=$\mbox{$x_{k-1}$}$]{xkm1} 
\Vertex[x=9, y=7, L=$\mbox{$x_k$}$]{xk} 
\Vertex[x=9, y=6, L=$\mbox{$x_{k+1}$}$]{xk1} 
\Vertex[x=9, y=5, L=$\mbox{$x_n$}$]{xn}
\Vertex[x=11, y=7, L=$\mbox{$\,$}$]{dummy}
\Edge(r)(x1)
\Edge(xkm1)(xk)
\Edge(xk)(xk1)
\Edge(xkm1)(dummy)
\tikzset{EdgeStyle/.style={dotted}}
\Edge(x1)(xkm1)
\Edge(xk1)(xn) 
\draw (11, 7) -- (10, 6) -- (12, 6) -- cycle;
\draw (11,5.6) node {subtree $z$};
\end{tikzpicture}
}}
\caption{\label{fig:lb-decoding}Left: hard instance in FCNS form, where $y$ is any tree of size $\Theta(n)$. Right: its decoded form.}
\end{figure}

Intuitively, decoding the tree of any hard instance requires to put the full tree $y$ into memory. Let $\XML(\FCNS(t))$ denote a hard instance which we aim to decode into $\XML(t)$. Then: $\XML(\FCNS(t)) = r x_{1L} \dots x_{nL} \cl{x_{nL}} \dots \cl{x_{k+1L}} Y \cl{x_{kL}} \dots \cl{x_{1L}} \cl{r_L}$ and the decoded form is $\mbox{$\XML(t) = r x_1 \dots x_n \cl{x_{n-1}} \dots \cl{x_k} Z \cl{x_{k-1}} \dots \cl{x_1} \cl{r} $}$ where $Z$ is the decoded form of $Y$.
Since $Z$ can only be written after $x_k$, and since $x_k$ cannot be memorized because $k$ was unknown until we reach $Y$,
memory space $\Omega(n)$ is required. This argument can be easily formalized using standard information theory arguments.
\begin{fact}\label{fact:lb-decoding}
Every one-pass randomized streaming algorithm for FCNS decoding
with bounded error requires space $\Omega(N)$. 
\end{fact}

This argument can be extended to two-pass randomized streaming algorithms. Construct a hard instance of size $\Theta(n^2)$ 
by gluing $n$ previous instances $(x^i,k^i,y^i)_{1\leq i\leq n}$ as follows: instance $(x^{i+1},k^{i+1},y^{i+1})$ is branched to the left most leaf of instance $(x^i,k^i,y^i)$. After the first pass, the algorithm is not able to write $n$ closing tags of form $\cl{x^i_{k_i}}$. Therefore he needs to store them in order to write them at the second pass.
This requires some formalization 
that we omit here. 
\begin{fact}\label{fact:lb-decoding2}
Every randomized streaming algorithm for FCNS decoding
with bounded error, one pass on the input stream, and two passes on the output stream 
requires space $\Omega(\sqrt{N})$. 
\end{fact}

\compress{\section*{Acknowledgements}
The authors  would like to thank Michel de Rougemont, who, among other things, introduced the authors to the problem of validating streaming XML documents.
}


{
\bibliographystyle{abbrv} \bibliography{km10} }


\appendix

\section{A $\Omega(N/p)$ space lower bound for $p$-pass Algorithms for $\pb$}
\label{sec:linear_space_lower_bound}
For the sake of clarity, in this section we provide a proof showing that $p$-pass algorithms require $\Omega(N/p)$ space for checking validity of arbitrary XML files against arbitrary DTDs. Many space lower bound proofs for Streaming Algorithms are reductions to problems in communication complexity \cite{ams99, Bar-YossefJKS04, mmn10}. For an introduction to communication complexity we refer the reader to \cite{nisan97}. 

Consider a player Alice holding an $N$ bit string $x=x_1 \dots x_N$, and a player Bob holding an $N$ bit string $y=y_1 \dots y_N$ both taken from a uniform distribution over $\{0, 1\}^N$. Their common goal is to compute the function $f(x, y) = \bigvee_i x[i] \wedge y[i]$ by exchanging messages. This communication problem is the well studied problem Set-Disjointness (\disj).

It is well known that the randomized communication complexity with bounded two-sided error of the Set Disjointness function $R(\disj) = \Theta(N)$. In this model, the players Alice and Bob have access to a common string of independent, unbiased coin tosses. The answer is required to be correct with probability at least $2/3$.

We make use of this fact by encoding this problem into an XML validity problem. Consider $\Sigma = \{r, 0, 1 \}$, the DTD $D^{\disj}$ such that $\mbox{$D^{\disj}(r) = 0r0 \, | \, 0r1 \, | \, 1r0 \, | \, \epsilon$}$, $D^{\disj}(0) = \epsilon$, and $D^{\disj}(1) = \epsilon$. Given an input $x, y$ as above, we construct an input tree $t(x,y)$ as in Figure~\ref{fig:hard_set_disjointness}.

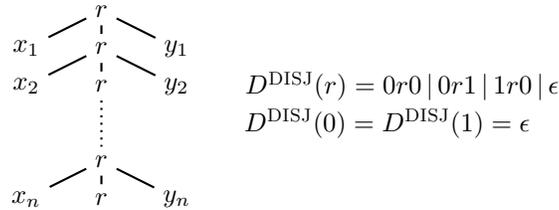
\begin{figure}[h]
\begin{center}
\begin{tikzpicture}[node distance = 0.5 cm]
\tikzstyle{VertexStyle} = [shape = rectangle, color = white, fill = white, text = black, draw]
\Vertex[x=10, y=10, L=$r$]{r} 
\Vertex[x=9, y=9.5, L=$x_1$]{a1}
\Vertex[x=10, y=9.5, L=$r$]{r1}
\Vertex[x=11, y=9.5, L=$y_1$]{b1}
\Vertex[x=9, y=9, L=$x_2$]{a2}
\Vertex[x=10, y=9, L=$r$]{r2}
\Vertex[x=11, y=9, L=$y_2$]{b2}
\Vertex[x=10, y=8, L=$r$]{rn}
\Vertex[x=9, y=7.5, L=$x_n$]{an}
\Vertex[x=10, y=7.5, L=$r$]{rn2}
\Vertex[x=11, y=7.5, L=$y_n$]{bn}
\Edge(r)(a1)
\Edge(r)(r1)
\Edge(r)(b1)
\Edge(r1)(a2)
\Edge(r1)(r2)
\Edge(r1)(b2)
\Edge(rn)(an)
\Edge(rn)(rn2)
\Edge(rn)(bn)
\tikzset{EdgeStyle/.style={dotted}}
\Edge(r2)(rn)
\draw (14,9) node {$D^{\disj}(r) = 0r0 \, | \, 0r1 \, | \, 1r0 \, | \, \epsilon$};
\draw (13.8,8.5) node {$D^{\disj}(0) = D^{\disj}(1) = \epsilon$};
\end{tikzpicture}
\end{center}
\caption{$t(x,y)$ is a hard instance for $\pb$. \label{fig:hard_set_disjointness} }
\end{figure}

Clearly, $\disj(x,y) = 0$ if and only if $\XML(t(x,y))$ is valid with respect to $D^{\disj}$.

\begin{theorem}
Every $p$-pass randomized streaming algorithm for $\pb$ with bounded error uses $\Omega(N/p)$ space, where $N$ is the input length.
\end{theorem}

\begin{proof}
Given an instance $x \in \{0, 1 \}^N$, $y \in \{0, 1 \}^N$ of $\disj$, we construct an instance for $\pb$. Then, we show that if there is a $p$-pass randomized algorithm for $\pb$ using space $s$ with bounded error, then there is a communication protocol for $\disj$ with the same error and communication $\Order(s \cdot p)$. This implies that any $p$-pass algorithm for $\pb$ requires space $\Omega(N/p)$ since $R(\disj) = \Theta(N)$.

Assume that $A$ is a randomized streaming algorithm deciding validity with space $s$ and $p$ passes. Alice generates the first half of $\XML(t(x,y))$, that is $rx_1\cl{x_1}rx_2\cl{x_2} \dots rx_n\cl{x_n}r$ of length $2N+2$ and executes algorithm $A$ on this sequence using a memory of size $\Order(s)$. Alice send the memory of size at most $s$ to Bob via message $M^1_A$ who continues algorithm $A$ on the second half of $\XML(t(x,y))$, that is $\cl{r}y_n\cl{y_n}\cl{r} \dots \cl{r}y_2\cl{y_2}\cl{r}y_1\cl{y_1}\cl{r}$ of length $2N+2$ using memory $M^1_A$. After execution, Bob sends the memory of size at most $s$ back to Alice via $M^1_B$. This procedure is repeated at most $p$ times. 

This protocol has a total length of $\Order(s \cdot p)$ which we know to be $\Omega(N)$ since $R(\disj) \in \Theta(N)$. The claim follows.
\end{proof}

\section{Conjecture: deciding $\pb(2)$ in $p$ passes requires $\Omega(\sqrt{N}/p)$ space}
\label{sec:root_space_lower_bound}

In the following, we motivate a conjecture lower-bounding the space for streaming algorithms deciding $\pb(2)$.

\begin{conjecture}
\label{conj:sqrt_space} A $p$-pass streaming algorithm for $\pb(2)$ requires $\Omega(\sqrt{N}/p)$ space.
\end{conjecture}

For positive integers~$m,n$, we define a family of hard instances for $\pb(2)$ of length $N=\Theta(mn)$ as in Figure~\ref{fig:lb-val-bin}.
The binary \ac{DTD} $D$ we consider is: $D(r) = 01 | 10, \  D(0) = 01 | 10 | 11 | \epsilon, \  D(1) = 01 | 10 | 00 | \epsilon$.

\begin{figure}[h] 
\centerline{\scalebox{0.7}{
\Tree [.$r$ [.$x[1]$ [.${\begin{array}{c} x[2]\\ \vdots\\ x[k]\end{array}}$ [.$x[k+1]$ [.${\begin{array}{c} x[k+2]\\ \vdots\\ x[n-1]\end{array}}$ [.$x[n]$ ] $\nega{x[n]}$ ] $\nega{x[k+2]}$ ] $d$ ]  $\nega{x[2]}$ ] $\nega{x[1]}$ ]   
}}
\centerline{\scalebox{0.7}{
\begin{tikzpicture}[node distance = 0.5 cm]
\tikzstyle{VertexStyle} = [shape = rectangle, color = white, fill = white, text = black, draw]
\Vertex[x=10, y=10, L=$r$]{r} 
\Vertex[x=8, y=9.5, L=$\mbox{$x_1 [1]$}$]{x11}
\Vertex[x=6, y=9, L=$\mbox{$x_1 [k_1]$}$]{x1k}
\Vertex[x=4, y=8.5, L=$\mbox{$x_1 [k_1+1]$}$]{x1k1}
\Vertex[x=2, y=8, L=$\mbox{$x_1 [n]$}$]{x1n}
\Vertex[x=7, y=8, L=$\mbox{$d_1$}$]{d1}
\Vertex[x=5, y=7.5, L=$\mbox{$x_2 [1]$}$]{x21}
\Vertex[x=3, y=7, L=$\mbox{$x_2 [k_2]$}$]{x2k}
\Vertex[x=1, y=6.5, L=$\mbox{$x_2 [k_2+1]$}$]{x2k1}
\Vertex[x=-1, y=6, L=$\mbox{$x_2 [n]$}$]{x2n}
\Vertex[x=4, y=6, L=$\mbox{$d_2$}$]{d2}
\Vertex[x=5, y=5, L=$\mbox{$\dots \dots$}$]{dots}
\Vertex[x=6, y=4, L=$\mbox{$x_m [1]$}$]{xm1}
\Vertex[x=4, y=3.5, L=$\mbox{$x_m [k_m]$}$]{xmk}
\Vertex[x=2, y=3, L=$\mbox{$x_m [k_m+1]$}$]{xmk1}
\Vertex[x=0, y=2.5, L=$\mbox{$x_m [n]$}$]{xmn}
\Vertex[x=5, y=2.5, L=$\mbox{$d_m$}$]{dm}
\Edge(r)(x11)
\Edge(x1k)(x1k1)
\Edge(x1k)(d1)
\Edge(d1)(x21)
\Edge(x2k)(x2k1)
\Edge(x2k)(d2)
\Edge(xmk)(xmk1)
\Edge(xmk)(dm)
\tikzset{EdgeStyle/.style=dotted}
\Edge(x11)(x1k)
\Edge(x1k1)(x1n)
\Edge(x21)(x2k)
\Edge(x2k1)(x2n)
\Edge(d2)(dots)
\Edge(dots)(xm1)
\Edge(xm1)(xmk)
\Edge(xmk1)(xmn)
\end{tikzpicture} }}
\caption{\label{fig:lb-val-bin} Top: instance for $m=1$, where $x$ is a $n$-bit string and $1\leq k\leq n-1$. Bottom: assembly of the base case to instances of size $m$. Instances are valid iff $(d_i\neq x_i[k] \text{ or }d_i\neq x_{i}[k+1])$, for all $1\leq i\leq m$.}
\end{figure}
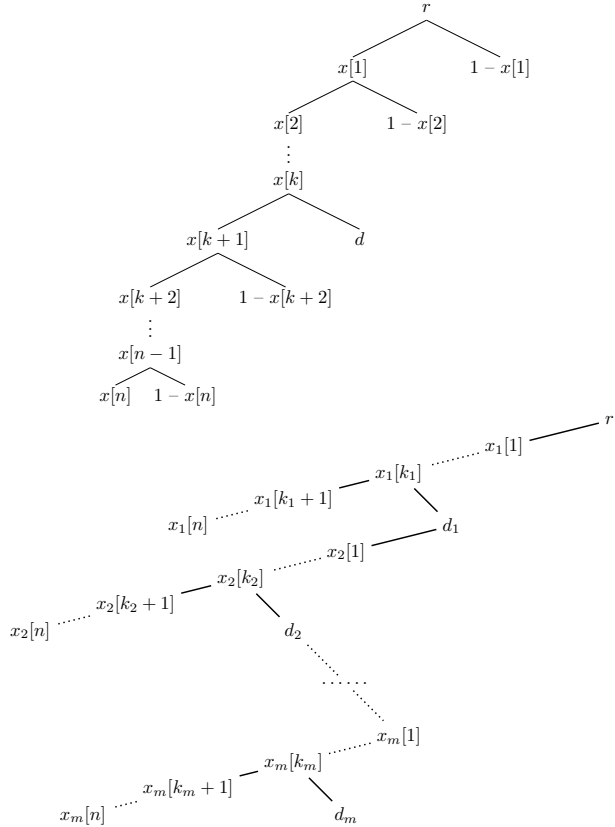

Intuitively, for $m=n$ deciding validity in one pass is difficult with space $\mathrm{o}(n)$. After reading the sequence of opening tags $x_i[1, n]$, the streaming algorithm does not have enough space to store the information about the bit at unknown index $k_i$. 
When it reads bit $d_i$, if $d_i=x_i[k_i+1]$, it is therefore unable to decide whether $d_i=x_i[k_i]$. 
Moreover, after reading $d_m$ it does not have enough space to store information
about all bits $x_1[k_1+1],x_2[k_2+1],\ldots ,x_m[k_m+1]$. 
When it reads the closing tags after $d_m$,  if $d_i=x_i[k_i]$,
it therefore misses out on its second chance to check whether $d_i=x_i[k_i+1]$ for every $i$.

Similar to \cite{mmn10, jn10, cckg10}, we translate this approach into the 
language of communication complexity. However, unlike in \cite{mmn10, jn10, cckg10} where the resulting communication problem is a two-party problem, the resulting communication problem involves three parties and comes with a
technical difficulty. We leave a proof for this lower bound as an open problem.

\end{document}